\documentclass[journal]{IEEEtran}

\usepackage{slashbox}

\usepackage{multicol}
\usepackage{multirow}
\usepackage{booktabs}
\usepackage{cite}
\ifCLASSINFOpdf
\else

\fi

\usepackage[cmex10]{amsmath}
\usepackage{amssymb}
\usepackage[linesnumbered,ruled,vlined]{algorithm2e}
\usepackage{mathrsfs}

\usepackage{array}
\usepackage{mdwmath}
\usepackage{mdwtab}

\newcommand{\ie}{\emph{i.e. }}

\usepackage[dvips]{graphicx}
\graphicspath{{eps/}}
\DeclareGraphicsExtensions{.eps,.jpg}

\IEEEoverridecommandlockouts
\newtheorem{theorem}{Theorem}
\newtheorem{lemma}[theorem]{Lemma}

\newtheorem{corollary}[theorem]{Corollary}
\newtheorem{observation}{Observation}
\newtheorem{definition}{Definition}
\newtheorem{example}{Example}

\begin{document}
%
\title{Binary Representaion for Non-binary LDPC Code with Decoder Design}
\author{{Yang Yu, Wen~Chen, Jun~Li, Xiao~Ma, and Baoming~Bai}
\thanks{Yang Yu, and Wen Chen are with
Network Coding and Transmission Laboratory, Shanghai Jiao Tong University,
Shanghai, 200240, China,
e-mail: \{yuyang83, wenchen\}@sjtu.edu.cn.}
\thanks{Jun Li is with School of Electrical and Information Engineering, The University of Sydney, NSW, Australia. e-mail: jun.li1@sydney.edu.au.}
\thanks{Xiao~Ma is with the Department of Electronics and Communication Engineering, Sun Yat-sen University, Guangzhou, China. e-mail: maxiao@mail.sysu.edu.cn.}
\thanks{Baoming Bai is with the State Key Lab of ISN, Xidian University, Xi’an, China.
e-mail: bmbai@mail.xidian.edu.cn.}
}

\markboth{Submitted to IEEE Transactions on Communications}%
{\MakeLowercase{Bit-Level Cycle Elimination for Non-binary and Binary LDPC Codes with Decoders Design}}

\maketitle

\begin{abstract}
The equivalent binary parity check matrices for the binary images of the cycle-free non-binary LDPC codes have numerous bit-level cycles.
In this paper, we show how to transform these binary parity check matrices into their cycle-free forms.
It is shown that the proposed methodology can be adopted not only for the binary images of non-binary LDPC codes but also for a large class of binary LDPC codes. Specifically, we present an extended $p$-reducible (EPR) LDPC code structure to eliminate the bit-level cycles.
For the non-binary LDPC codes with short length symbol-level cycles, the EPR-LDPC codes can largely avoid the corresponding short length bit-level cycles. As to the decoding of the EPR-LDPC codes, we propose a hybrid hard-decision decoder and a hybrid parallel decoder for binary symmetric channel and binary input Gaussian channel, respectively. A simple code optimization algorithm for these binary decoders is also provided. Simulations show the comparative results and justify the advantages, i.e., better performance and lower decoding complexity, of the proposed binary constructions.
\end{abstract}


%
\IEEEpeerreviewmaketitle
\section{Introduction} 
\label{sec:introduction}
Low density parity check (LDPC) codes, as a class of forward error control codes, have gained considerable attentions during the last decade due to their amazing decoding performance under different channels~\cite{Rchardson01DensityEvolution,Davey98montecarlo}. The performance of an infinite long LDPC code is usually evaluated in terms of the threshold for the average performance of its code ensemble (codes with the same degree distributions) based on the cycle-free condition~\cite{Rchardson01DensityEvolution,Brink04Exit,Brannstrom05Convergence,ge09NLDPC,Sassatelli10Hybrid,sv08DensityEvolution}.

The LDPC codes will suffer from performance degradation if there exist non-negligible number of short length cycles in their parity check matrices, especially for the short block length codes.
Moreover, codes with large girths will have respectable minimum/stopping distance bound, which also implies enhanced decoding performance.
In this paper, we refer to the cycles in the binary parity check matrices as \emph{bit-level} cycles and the cycles in the non-binary parity check matrices as \emph{symbol-level} cycles.
In \cite{Yige2013Girth,Christian2009QcGirth,Guohua2010QcGirth}, the authors show how to construct the parity check matrices with less bit-level cycles and large girths for binary LDPC codes.
For the non-binary LDPC codes, investigations indicate that they could have sparser Tanner graphs as the field size increases. For short to moderate block lengths, the non-binary LDPC codes with sparser graphs are more likely to outperform the binary ones.
In~\cite{Huang09Cycle,Poulliat08BinImage}, the authors investigate a particular type of non-binary LDPC codes, \ie non-binary cycle LDPC codes, whose column weights are two. In~\cite{Huang09Cycle}, optimizations for this type of codes are performed over Cayley-graph. In~\cite{Poulliat08BinImage}, the authors propose bit-level coefficients selection methods to optimize the symbol-level performance for the non-binary cycle LDPC codes.

On the other hand, soft-decision decoding for the non-binary LDPC codes requires a potentially higher complexity. The complexity of the $q$-ary sum-product decoding algorithm (QSPA) is $O(q^2)$ for each check-sum operation. The Fourier transform QSPA reduces the complexity to $O(q\log q)$~\cite{ge09NLDPC}. The extended min-sum (EMS) algorithm in \cite{Declercq2007EMS} further reduces the complexity to $O(n_m\log n_m)$ at the cost of a bit performance loss, where $n_m$ is smaller than $q$.
However, the computational complexity of the EMS decoder is still very high compared to the binary decoder.
Hence, in~\cite{Savin09Erasure,SV11ExtendedErasure}, the authors propose an extended binary representation for the non-binary LDPC code which can be decoded by binary decoders. The binary decoding complexity is only $O(q)$ for BEC.
Theoretically, based on the decoding error probability, the authors in~\cite{smith10Pct,Yang12ApproximatedPCT} prove that the minimal decoding complexities exist if the LDPC codes are constructed with properly chosen degree distributions.

\subsection{Related Works} 
\label{sub:related_works}
The codewords of a non-binary LDPC code are often transmitted over binary input channels in their bit-vector forms, i.e., binary images of the non-binary LDPC codes. At the receiver side, the non-binary decoder needs to transform the received bit sequences back to their non-binary forms to perform the symbol-level decoding\cite{Davey98montecarlo,Sassatelli10Hybrid,Voicila2008SplitCode,Poulliat08BinImage,Xiuni09FastGldpc} for retrieving the information bits.
On the other hand, as an alternative of using the non-binary decoders for binary input channels, one can use a binary decoder to retrieve the information bits by utilizing the binary representations of the non-binary parity check matrices for the purpose of reducing the computational complexity\cite{Savin09Erasure,SV11ExtendedErasure,Yang12Erasure}.
Especially in certain cases, when the receiver receives a non-binary codeword (e.g., moderate to long block length where the non-binary decoders do not have a clear advantage compared to the binary decoders) from the binary input channels and only limited computational resources are available, the consideration of using binary decoders is natural and practical for a fast and correct information recovery.
However, the binary representation of a non-binary parity check matrix has numerous bit-level cycles, even if there is no symbol-level cycle~\cite{Yang12Erasure,Savin09Erasure} in the non-binary parity check matrix.
Thus, in~\cite{Savin09Erasure,SV11ExtendedErasure}, the authors introduce the (punctured) extended binary representation for the non-binary LDPC code to solve this issue. When there is no symbol-level cycle, this representation will also be cycle-free.
In~\cite{Yang12Erasure}, the authors propose a hybrid hard decision decoder particularly for the BEC which eliminates the local decoding cycles by introducing matrix inverse operations. Decoding in~\cite{Yang12Erasure} has lower computational complexity than the decoding of the extended binary representation.
In addition, the authors in~\cite{Bhatia2011Extended} show how to optimize the binary representation of a non-binary parity check matrix with the perspective of stopping set.

\subsection{Contributions} 
\label{sub:contributions}
In this paper, we aim at further improving the bit-level decoding performance and reducing the bit-level decoding complexity. We propose bit-level decoders for different channel models to achieve enhanced decoding performance and develop a new methodology to construct the generalized binary representation to avoid the short length bit-level cycles.
Specifically, we propose a hybrid hard-decision decoder and a hybrid parallel decoder for binary symmetric channel and binary input Gaussian channel, respectively.
We also develop an extended $p$-reducible (EPR) LDPC code structure for a large class of LDPC codes with decoding complexity of $O(m_s),m_s<q$. For the non-binary LDPC codes with short length symbol-level cycles, the EPR-LDPC codes can largely avoid the corresponding bit-level cycles. Experimental studies show that the proposed EPR-LDPC codes under the hybrid parallel decoder have a maximum $0.8$dB performance gain compared to the optimized non-binary cycle LDPC codes~\cite{Huang09Cycle,Poulliat08BinImage,Voicila2008SplitCode} with a much lower decoding complexity.

Contributions of this paper are summarized as follows.
\begin{enumerate}
  \item We propose an extended iterative hard decision decoder and a hybrid parallel decoder for different channel models. A simple code optimization algorithm for these binary decoders is also provided to guarantee enhanced decoding performance.
  \item We propose an EPR-LDPC code structure to avoid the short length bit-level cycles. A general framework is given to model the constructions and optimizations of the EPR-LDPC codes. Significant results and conditions regarding the constructions and optimizations of the EPR-LDPC codes are also derived.
\end{enumerate}

\subsection{Organization of the paper} 
\label{sub:organization_of_the_paper}
The contents of this paper are organized as follows.
In section~\ref{sec:Preliminaries}, we introduce the binary representations of the non-binary LDPC code and give a unified framework for the extended binary representation.
In section~\ref{sec:EPR-LDPC}, we give the details about the EPR-LDPC codes.
In section~\ref{sec:Class_Faster_Decoders}, we give the proposed binary decoders and provide a simple code optimization algorithm.
Section~\ref{sec:Simulation} presents the simulation results.
\section{Binary Representations for Non-binary LDPC Codes} 
\label{sec:Preliminaries}
\subsection{Binary Images for Non-binary LDPC Codes} 
\label{sub:Bin_form}
Let $\mathbb F_q$ be the finite field of size $q=2^p$. Let $\mathbb F_q^*=\mathbb F_q\backslash \{0\}$ and $\mathbb F_q^N$ be the $q$-ary column vector space of dimension-$N$. The non-binary LDPC code $\mathcal C$ of length $N$ is the dimension $N-M$ linear subspace of $\mathbb{F}_q^N$. Let $\mathbf H=\{h_{i,j}\}_{M\times N},h_{i,j}\in \mathbb{F}_q$ be the parity check matrix. Then the non-binary LDPC code $\mathcal{C}$ is defined as the kernel of $\mathbf{H}$. If we endow $\mathbb F_q$ with a binary vector space structure, every $u\in\mathbb F_q$ can be denoted by a binary vector $\bar{\mathbf{u}}=(\bar{u}_{1},\bar{u}_{2},\ldots,\bar{u}_{p-1})^T$. As a result, each codeword $\mathbf x = (x_1,x_2,\ldots,x_N)^T\in \mathbb{F}^N_q$ in $\mathcal C$ has its \emph{binary vector representation} \[\bar{\mathbf x} = (\bar{\mathbf x}_1^T,\bar{\mathbf x}_2^T,\ldots,\bar{\mathbf x}_N^T)^T,\] i.e., \emph{binary images} of the non-binary LDPC codes.

To obtain the binary representation of $\mathbf H$, we use the \emph{companion matrix} $\mathbf A$ over $\mathbb F_q$ \cite{Lidl86FiniteFields,xiao07UniDecLdpc,Yang12Erasure}.
Then we have $\mathbb F_q \cong\{\mathbf{0},\mathbf A^i,0\leqslant i\leqslant q-2\}$.
If we replace every $h_{i,j}$ in $\mathbf H$ by its binary matrix representation $\mathbf A_{i,j}$ (also called \emph{matrix label}), we get the \emph{equivalent binary parity check matrix} \[\mathbf{\bar{H}}=\{\bar{\mathbf H}_1,\bar{\mathbf H}_2,\dots,\bar{\mathbf H}_M\}^T,\] where $\bar{\mathbf H}_i,i=1,\ldots,M$ are $Np\times p$ matrices. As a result, the \emph{equivalent binary LDPC code} $\bar{\mathcal C}$ is defined as the kernel of the matrix $\mathbf{\bar{H}}$.

With a little abuse of the notation, in the following, we denote any binary parity check matrix by $\bar{\mathbf H}$ and any non-binary parity check matrix by ${\mathbf H}$.
We define $\mathrm{diag}(\mathbf B_1,\mathbf B_2,\ldots,\mathbf B_N)$ as the matrix
\begin{equation}
  \nonumber
  \mathrm{diag}(\mathbf B_1,\mathbf B_2,\ldots,\mathbf B_N)=
  \left(\begin{array}{cccc}
  \mathbf B_1 & \mathbf 0 & \cdots & \mathbf 0\\
  \mathbf 0 & \mathbf B_2 & \cdots & \mathbf 0 \\
  \vdots & \vdots & \ddots & \vdots \\
  \mathbf 0 & \mathbf 0 & \cdots & \mathbf B_N \\
\end{array}\right),
\end{equation}
where $\mathbf B_j,j=1,2,\ldots,N$, are not necessarily to be square matrices.

We also define $\otimes$ as the kronecker product of two matrices.
Let $\mathbf B=(b_{i,j})_{m\times n}$ and $\mathbf B^{\prime}=(b^{\prime}_{i,j})_{h\times g}$ be two arbitrary matrices.
The kronecker product of $\mathbf B$ and $\mathbf B^{\prime}$ is an $mh\times ng$ matrix
\begin{equation}
  \nonumber
  \mathbf B\otimes\mathbf B^{\prime}=
  \left(\begin{array}{cccc}
  b_{1,1}\mathbf B^{\prime} & b_{1,2}\mathbf B^{\prime} & \cdots & b_{1,n}\mathbf B^{\prime}\\
  b_{1,1}\mathbf B^{\prime} & b_{2,2}\mathbf B^{\prime} & \cdots & b_{2,n}\mathbf B^{\prime} \\
  \vdots & \vdots & \ddots & \vdots \\
  b_{m,1}\mathbf B^{\prime} & b_{m,2}\mathbf B^{\prime} & \cdots & b_{m,n}\mathbf B^{\prime} \\
\end{array}\right).
\end{equation}
\subsection{Extended Binary Representation for Non-Binary LDPC Codes} 
\label{sub:Extended_form}
In this subsection, we give a unified framework for the extended binary representation.
Let $\mathbb N$ be the set of natural integers including $0$ and $\mathbb N^*=\mathbb N\backslash\{0\}$.
Let $\mathbb N_{q}=\{0,1,\ldots,q-1\}$ and $\mathbb N^*_q=\mathbb N_q\backslash \{0\}$.
For an arbitrary matrix $\mathbf B$, we denote the entries of $\mathbf B$ by $\mathbf B(i,j),i,j\in \mathbb N$, where $i$ and $j$ are the row number and column number, respectively. In addition, $\mathbf B(i,0)$ represents the $i^{th}$ row vector, $\mathbf B(0,j)$ represents the $j^{th}$ column vector.
Let $\mathbf I_{p\times p}$ be the $p\times p$ identity matrix.
The extended representation is based on the equivalent binary LDPC code, which begins with a linear transformation of a binary vector $\bar{\mathbf x}_j\in\mathbb F_2^p$\cite{Savin09Erasure}.

Let $\mathbf \Phi$ be the $p\times (q-1)$ binary matrix of the following form \[\mathbf \Phi=(\mathbf{\Phi}(0,1),\mathbf{\Phi}(0,2),\ldots,\mathbf{\Phi}(0,q-1)),\] where each column vector $\mathbf{\Phi}(0,j),j=1,2,\ldots,q-1$, is the binary representation of $j\in\mathbb N_q^*$. Let $\bar{\mathbf x}_j$ be the binary vector representation of the coded symbol $x_j$ and $\mathbf v_j=\mathbf{\Phi}^T\bar{\mathbf x}_j\in\mathbb F_2^{q-1}$.
Note that $\mathbf\Phi$ is the parity check matrix of the $[q-1,q-1-p]$ hamming code. So, each $\mathbf v_j$ is also a codeword of the simplex code (dual code of the hamming code).
The \emph{extended binary representation} of $\mathbf x$ is then \[\mathbf v=(\mathbf v_1^T,\ldots,\mathbf v_N^T)^T.\] In addition, for each non-zero $\mathbf A_{i,j}$, we can get a $(q-1)\times(q-1)$ matrix $\mathbf\Omega_{i,j}$ by utilizing an endomorphism of $\mathbb N_q$ and an isomorphism between $\mathbb N_q$ and $\mathbb F_2^p$\cite{Savin09Erasure}. If we replace the non-zero $\mathbf A_{i,j}$ in $\mathbf{\bar{H}}$ by $\mathbf\Omega_{i,j}$ and the zero $\mathbf A_{i,j}$ by $\mathbf 0_{(q-1)\times(q-1)}$, we get the extended binary parity check matrix $\mathbf \Omega=(\mathbf\Omega_{i,j})_{M\times N}$. Then $\mathbf\Omega \mathbf v = \mathbf 0$ and the simplex constraints on $\mathbf v$ together form the extended binary representation.

In section~\ref{sec:EPR-LDPC}, we will introduce a matrix map $f_\omega$ for constructing and optimizing the parity check matrix of the EPR-LDPC codes. We also notice that the matrix map can be utilized to simplify the constructions of $\mathbf\Omega$. As a result, $\mathbf\Omega_{i,j}=f_{\omega}(\mathbf\Phi,\mathbf A_{i,j})$. Details about $f_\omega$ is given in section~\ref{sec:EPR-LDPC}. Here, we only present the unified framework for the extended binary representation.

\section{Extended $p$-reducible Code} 
\label{sec:EPR-LDPC}
In this section, we give the general framework to model the constructions and optimizations of the EPR-LDPC codes and show how to design EPR-LDPC codes by satisfying some girth constraints.
\subsection{Definition of the EPR-LDPC Codes} 
\label{sub:definition_of_the_epr_ldpc_codes}
We first give the definition of $p$-reducible codes.
\begin{definition}[$p$-reducible]\label{def:p_reducible}
  Given a binary parity check matrix $\bar{\mathbf H}$ and an integer $p>1$. The binary code defined with $\bar{\mathbf H}$ is called \emph{$p$-reducible} iff, after rearranging the columns and rows of $\bar{\mathbf H}$, $\bar{\mathbf H}$ can be expressed as $(\mathbf A_{i,j})_{M\times N}$.
  Each $\mathbf A_{i,j}$ is either $p\times p$ zero matrix or $p\times p$ full-rank matrix.
  This binary code defined with $\bar{\mathbf H}$ is called \emph{strictly $p$-reducible} iff there does not exist an integer $p^\prime\neq p$ such that the binary code defined with $\bar{\mathbf H}$ is also $p^\prime$-reducible.
\end{definition}

From Definition~\ref{def:p_reducible}, we know that the equivalent binary LDPC codes are $\log_2 q$-reducible.
The length-12 (3,4)-regular Gallager code in\cite{JohnsonLDPCintro} is a strictly 3-reducible code.
For the quasi-cyclic binary LDPC codes, we notice that a circulant is full-rank iff its associated polynomial\cite{London1956RankCurculant} and $1-x^n$ has only one common zero.
Then if the $(Np,(N-M)p)$ binary quasi-cyclic LDPC codes are constructed with $p\times p$ full rank circulant matrices,
these codes are $p$-reducible.
In\cite{Lin2004Error}, the authors show that the parity check matrices of irregular repeat accumulate (IRA) codes can be constructed by an array of some circulant permutation matrices.
The resulting parity check matrices are composed of identity matrices and the circulant permutation matrices.
Let the size of the circulant permutation matrices be $p\times p$.
Then these IRA codes are $p$-reducible.
Moreover, if the protograph LDPC codes are obtained by filling the base matrices with $p\times p$ full rank matrices and zero matrices, the resulting codes are also $p$-reducible codes.
The above examples are only a small list of the $p$-reducible codes.

Next, we give the definitions that will be used in the following sections.
\begin{definition}
  The \emph{mother matrix} $\mathbf\Lambda_p$ of a binary matrix $\bar{\mathbf H}$ or of a non-binary matrix $\mathbf H$ is defined as a matrix with each entry being either $0$ or $1$.
  The binary matrix $\bar{\mathbf H}$ can be obtained by replacing the 0s by $\mathbf 0$ matrices of size $p\times p$ and the 1s by non-zero matrices of size $p\times p$.
  These $p\times p$ matrices are also referred to as the \emph{matrix labels}.
  The non-binary matrix $\mathbf H$ can be obtained by replacing the 0s in $\mathbf\Lambda_p$ by the zero element in $\mathbb F_{2^p}$ and the 1s by the non-zero elements in $\mathbb F_{2^p}$.
  Cycles in $\mathbf\Lambda_p$ or $\mathbf H$ are referred to as the \emph{symbol-level} cycles.
  Cycles in $\bar{\mathbf H}$ are referred to as the \emph{bit-level} cycles.
\end{definition}

\begin{definition}\label{def:pre_EPR}
	Let $\bar{\mathbf H}=(\mathbf A_{i,j})_{M\times N}=(\bar{\mathbf H}^c_j)_{1\times N}$ be a binary parity check matrix.
	Let $\mathbf\Psi_j,j\in\{1,\ldots,N\}$ be a $p\times p$ full-rank matrix. Let $\hat{\mathbf\Psi}^e_j=\{\mathbf\Psi^{e^1}_j,\mathbf\Psi^{e^2}_j,\ldots,\mathbf\Psi^{e^{M(q-1)}}_j\},j\in\{1,\ldots,N\}$ be a matrices set, where $\mathbf\Psi^{e^{i_q}}_j,i_q=1,2,\ldots,M(q-1)$ are matrices with $p$ rows and $p_j^\prime\leqslant q-1$ columns.
	Let $f_\omega$ be a matrix map and $f_\omega(\mathbf\Psi^{e^{i_q}}_j,\mathbf\Psi_j)$ be a matrix with $p_j^\prime$ columns. If each non-zero $\mathbf A_{i,j}$ is a full rank matrix, we define $f_e$ as a function that takes $\bar{\mathbf H}$ and $\hat{\mathbf\Psi}^{e}_j,j=1,\ldots,N$ as inputs and outputs $\mathbf\Omega^e=(\mathbf\Omega^e_{i,j})_{M\times N}$, where $\mathbf\Omega^e_{i,j}=\sum_{i_q=(i-1)(q-1)+1}^{i(q-1)} f_{\omega}(\mathbf\Psi^{e^{i_q}}_j,\mathbf A_{i,j})$.
\end{definition}

Below, before the detailed constructions, we give the definition of the EPR-LDPC codes to provide some general ideas.
\begin{definition}[EPR-LDPC]\label{def:EPR}
  Let $\mathbf\Psi^e=\{\mathbf\Psi^e_j,j=1,2,\ldots,N\}$ be the \emph{extended generator matrices set}.
  Each $\mathbf\Psi^e_j$ is a full-rank binary matrix with $p$ rows and $p_j^\prime\leqslant q-1$ columns and the non-zero columns in each $\mathbf\Psi^e_j$ are different from each other.
  Let \[\mathbf v^e = \mathrm{diag}(\mathbf\Psi^{eT}_1,\mathbf\Psi^{eT}_2,\ldots,\mathbf\Psi^{eT}_N)\cdot\bar{\mathbf x},\]
  where $\mathbf v^e =(\mathbf v^{eT}_1,\mathbf v^{eT}_2,\ldots,\mathbf v^{eT}_N)^T$
  and $\bar{\mathbf x}$ is the binary codeword of $\bar{\mathbf H}$.
  We associate each $\bar{\mathbf H}^c_j$ with a matrix set $\hat{\mathbf\Psi}^e_j$.
  Then, the EPR-LDPC code is defined as the kernel of the parity check matrix
  \[\mathbf\Omega^e=f_e(\bar{\mathbf H},\hat{\mathbf\Psi}^{e}_1,\hat{\mathbf\Psi}^{e}_2,\ldots,\hat{\mathbf\Psi}^{e}_N),\]
  such that $\mathbf\Omega^e\cdot\mathbf v^e=\mathbf 0$ and the matrices in $\hat{\mathbf\Psi}^e_j$ has the same column number as $\mathbf\Psi^e_j$.
\end{definition}

The above definition is very broad.
It does not only defines more generalized binary representations for the non-binary LDPC codes, but also defines the generalized representations for a large class of binary LDPC codes.
All the $p$-reducible codes can be connected to the EPR-LDPC codes.
\subsection{Mapping definition and Examples} 
\label{sub:mapping_definition_and_examples}
In this subsection, we give the details about the matrix map $f_\omega$.
We begin with the basic notations and definitions.
Let $\mathrm{wt}(\cdot)$ be the function that calculates the number of non-zero columns in a matrix or of the non-zero elements in a vector.
\begin{definition}[$\preccurlyeq$]\label{def:Preccurlyeq}
    We denote the relationship between two vectors $\mathbf a,\mathbf b$ by $\mathbf a\preccurlyeq \mathbf b$ if $\mathbf a$ is obtained by replacing some elements in $\mathbf b$ by zeros.
    For two matrices $\mathbf A,\mathbf B$, we denote $\mathbf A\preccurlyeq \mathbf B$ if $\mathbf A$ is obtained by replacing some column vectors in $\mathbf B$ by zero vectors.
\end{definition}
\begin{definition}[$\prec$]\label{def:Prec}
    We denote the relationship between two vectors $\mathbf a,\mathbf b$ by $\mathbf a\prec \mathbf b$ if $\mathbf a\preccurlyeq \mathbf b$ and $\mathrm{wt}(\mathbf a)<\mathrm{wt}(\mathbf b)$.
    For two matrices $\mathbf A,\mathbf B$, we denote the relationship between them by $\mathbf A\prec \mathbf B$ if $\mathbf A\preccurlyeq \mathbf B$ and $\mathrm{wt}(\mathbf A)<\mathrm{wt}(\mathbf B)$.
\end{definition}

Note that $\mathbf\Psi^e_j$ has different non-zero vectors as its columns and $\mathbf\Phi$ has all the non-zero vectors in $\mathbb F_2^p$ as its columns. The non-zero column vectors in each $\mathbf\Psi^e_j$ form a subset of the column vectors in $\mathbf\Phi$. Without loss of generality, we assume that $\mathbf\Psi^e_j\preccurlyeq\mathbf\Phi$ for all $j\in\{1,2,\ldots,N$\}. Since the zero columns in $\mathbf\Psi^e_j$ will result in zero bits in $\mathbf v^e_j$ which can be ignored or readily removed, this assumption does not violate the Definition~\ref{def:pre_EPR}-\ref{def:EPR} and will facilitate the discussion of EPR-LDPC code too.
With a little abuse of notation, we use $f_{\omega}(\mathbf\Psi^{e^{i_q}}_j,\mathbf\Psi_j)$ to denote the resulting binary matrix
and $f_{\omega,j}(i^\prime,j^\prime),i^\prime,j^\prime\in\{1,2,\ldots,q-1\}$ to denote the entries in $f_{\omega}(\mathbf\Psi^{e^{i_q}}_j,\mathbf\Psi_j)$. Then
\begin{equation} \nonumber
  f_{\omega,j}(i^\prime,j^\prime)=
  \left\{
  \begin{aligned}
    1, \text{     if } \mathbf\Psi^{e^{i_q}}_j(0,j^\prime)+\mathbf\Psi_j^T\mathbf\Phi(0,i^\prime)=\mathbf 0, \\
    0, \text{     if } \mathbf\Psi^{e^{i_q}}_j(0,j^\prime)+\mathbf\Psi_j^T\mathbf\Phi(0,i^\prime)\neq\mathbf 0.
  \end{aligned}
  \right.
\end{equation}
If we replace $\mathbf\Psi_j$ by the matrix label $\mathbf A_{i,j}$, then different columns of $f_{\omega}(\mathbf\Psi^{e^{i_q}}_j,\mathbf A_{i,j})$ associate with different bits in $\mathbf v^e_j$. Different rows of $f_{\omega}(\mathbf\Psi^{e^{i_q}}_j,\mathbf A_{i,j})$ represent different additions between the bits in $\mathbf v^e_j$.
To have a better understanding, we give simple examples for $f_\omega$.

\begin{figure}[!hbtp]
	\begin{center}
	\includegraphics[width=.5\textwidth]{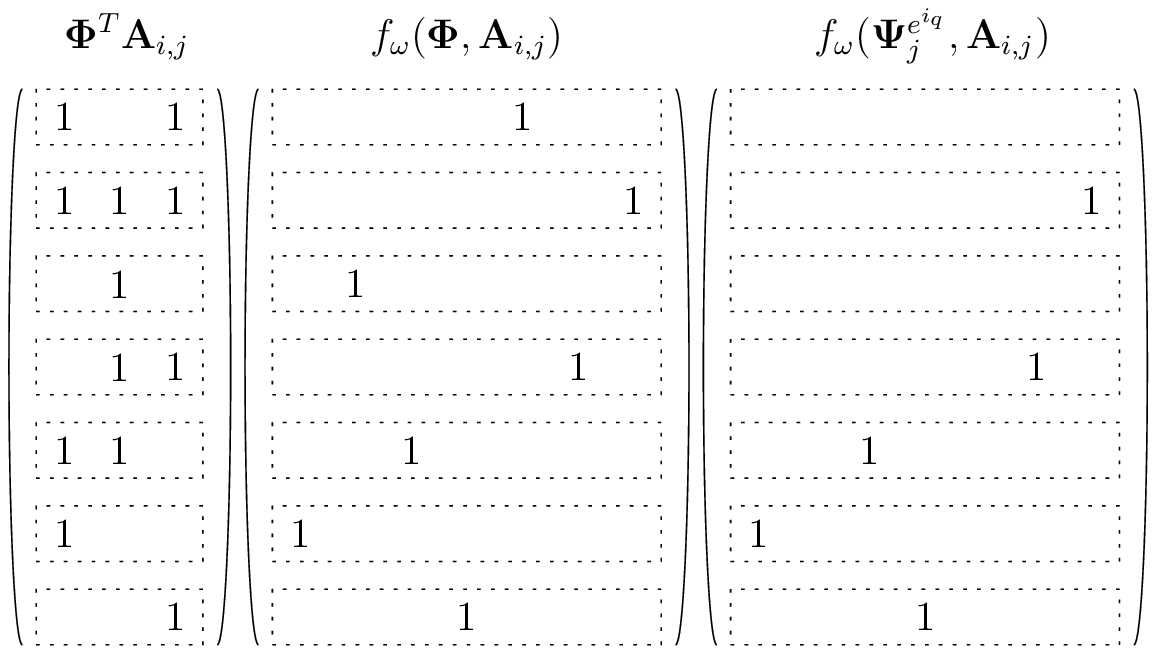}
	\end{center}
	\vspace*{-.3cm}
	\caption{Different matrices generated by $f_\omega$ in Example~\ref{exam:fw_Omega}.}
	\label{fig:exm_1}\vspace*{-.3cm}
\end{figure}

\begin{example}\label{exam:fw_Omega}
	Since the extended generator matrix $\mathbf\Psi^e_j$ has different non-zero columns, the corresponding bits in $\mathbf v^e_j$ will represent different combinations of the bits in $\bar{\mathbf x}_j$.
	Moreover, the additions between different binary parity check equations within $\bar{\mathbf H}_i^T\bar{\mathbf x}=\mathbf 0,i\in\{1,2,\ldots,M\}$ can be formulated as $\mathbf \Phi^T \bar{\mathbf H}_i^T\bar{\mathbf x}=\mathbf 0$ which will result in $q-1$ different binary parity check equations\cite{Savin09Erasure,Savin2012Fourier}.
	We divide the $q-1$ binary parity check equations into $N$ partitions with the $j^{th}$ partition consisting of
	different combinations of the bits in $\bar{\mathbf x}_j$, i.e., $\mathbf \Phi^T {\mathbf A}_{i,j}\bar{\mathbf x}_j$.
	If we set some of the $q-1$ equations to be zero equations,
	then there exist only one $\mathbf\Psi^{e^{i_q}}_j$ for the $j^{th}$ partition such that the $q-1$ rows of $f_{\omega}(\mathbf\Psi^{e^{i_q}}_j,\mathbf A_{i,j})$ respectively represents the $q-1$ rows within the $j^{th}$ partition,
	e.g., if $p=3$ and $\mathbf A_{i,j}=(\mathbf\Phi(0,3),\mathbf\Phi(0,6),\mathbf\Phi(0,7))$, then $f_{\omega}(\mathbf\Phi,\mathbf A_{i,j})$=$\mathbf\Omega_{i,j}$.
	If we set the first and third rows in $\mathbf\Omega_{i,j}$ to be zero vectors, then we have
	\begin{eqnarray}\nonumber
		\mathbf\Psi^{e^{i_q}}_j&=&(\mathbf\Phi(0,1),\mathbf 0,\mathbf\Phi(0,3),\mathbf\Phi(0,4),\mathbf 0,\\\nonumber
		&&\mathbf\Phi(0,6),\mathbf\Phi(0,7),\mathbf\Phi(0,8))\prec\mathbf\Phi,\\\nonumber
		f_{\omega}(\mathbf\Psi^{e^{i_q}}_j,\mathbf A_{i,j})&=&(\mathbf 0,\mathbf\Omega_{i,j}(2,0)^T,\mathbf 0,\mathbf\Omega_{i,j}(4,0)^T,\mathbf\Omega_{i,j}(5,0)^T,\\\nonumber
		&&\mathbf\Omega_{i,j}(6,0)^T,\mathbf\Omega_{i,j}(7,0)^T)^T\prec\mathbf\Omega_{i,j}.
	\end{eqnarray}
	More details are displayed in Fig~\ref{fig:exm_1}.
\end{example}

Note that each $\mathbf v^e_j$ is a codeword generated by $\mathbf\Psi^e_j$ and a combination of the bits in $\bar{\mathbf x}_j$ can be represented as a non-zero entry in $f_{\omega}(\mathbf\Psi^{e^{i_q}}_j,\mathbf A_{i,j})$.
$f_\omega$ can be also used to represent some parity check relationships for the bits in one $\mathbf v^e_j$.
The construction of such matrices is trivial, so we leave it for briefness.
Moreover, different additions of the rows of $\bar{\mathbf H}$ and the combinations of the parity check relationships for each $\mathbf v^e_j$ can all be represented by $f_\omega$.
\subsection{Main Properties of the Mapping} 
\label{sub:main_properties_of_the_mapping}
\begin{lemma}\label{lemma:Property_fw}
  Let $\mathbf B\preccurlyeq\mathbf\Phi$ and $\mathbf B^\prime\preccurlyeq\mathbf\Phi$ be two $p\times (q-1)$ binary matrices.
  Let $\mathbf\Psi_j$ be a $p\times p$ full-rank binary matrix.
  $f_\omega(\mathbf\Phi,\mathbf\Psi_j)$ is a $(q-1)\times (q-1)$ permutation matrix.
  In addition, $\mathbf B^\prime\preccurlyeq\mathbf B$ and $f_\omega(\mathbf B^\prime,\mathbf\Psi_j)\preccurlyeq f_\omega(\mathbf B,\mathbf\Psi_j)$ are necessary and sufficient conditions for each other.
\end{lemma}
\begin{proof}
  Since $\mathbf\Psi_j$ is a $p\times p$ full rank matrix, all the $\mathbf\Psi^T_j\mathbf\Phi(0,i^\prime),i^\prime=1,2,\ldots,q-1$ are different column vectors.
  Then $f_\omega(\mathbf\Phi,\mathbf\Psi_j)$ will have only one non-zero entry in each row or column.
  So, $f_\omega(\mathbf\Phi,\mathbf\Psi_j)$ is a $(q-1)\times (q-1)$ permutation matrix.
  If $\mathbf B\preccurlyeq\mathbf\Phi$, the zero columns in $\mathbf B$ will result in zero rows in $f_\omega(\mathbf B,\mathbf\Psi_j)$.
  Then $f_\omega(\mathbf B,\mathbf\Psi_j)$ can be obtained by setting some rows of $f_\omega(\mathbf\Phi,\mathbf\Psi_j)$ to be zero vectors.
  Since $f_\omega(\mathbf\Phi,\mathbf\Psi_j)$ have only one non-zero entry in each column, then some columns become zero vectors in $f_\omega(\mathbf B,\mathbf\Psi_j)$.
  As a result $f_\omega(\mathbf B,\mathbf\Psi_j)\preccurlyeq f_\omega(\mathbf\Phi,\mathbf\Psi_j)$.
  Similarly, we have $f_\omega(\mathbf B^\prime,\mathbf\Psi_j)\preccurlyeq f_\omega(\mathbf B,\mathbf\Psi_j)$ if $\mathbf B^\prime\preccurlyeq\mathbf B$.
  Conversely, if $f_\omega(\mathbf B,\mathbf\Psi_j)\preccurlyeq f_\omega(\mathbf\Phi,\mathbf\Psi_j)$, it means that
  the columns in $\mathbf B$ generating the zero rows in $f_\omega(\mathbf B,\mathbf\Psi_j)$ are set to be zero vectors.
  Since there is a one-to-one correspondence between $\mathbf B$ and $f_\omega(\mathbf B,\mathbf\Psi_j)$, $\mathbf B\preccurlyeq\mathbf\Phi$.
  Similarly, we have $\mathbf B^\prime\preccurlyeq\mathbf B$ if $f_\omega(\mathbf B^\prime,\mathbf\Psi_j)\preccurlyeq f_\omega(\mathbf B,\mathbf\Psi_j)$. This completes the proof.
\end{proof}

In the following, if each $\mathbf A_{i,j}$ in $\bar{\mathbf H}$ is replaced by $f_{\omega}(\mathbf\Phi,\mathbf A_{i,j})$, we denote the resulting matrix by
\begin{eqnarray}\nonumber
	\mathbf\Omega&=&(\mathbf\Omega_{i,j})_{M\times N}\\\nonumber
	&=&(\mathbf\Omega_1,\mathbf\Omega_2,\ldots,\mathbf\Omega_M)^T=(\mathbf \Omega_1^c,\mathbf \Omega_2^c,\ldots,\mathbf \Omega_N^c),
\end{eqnarray}
where $\mathbf \Omega_i$ is the $(q-1)N\times (q-1)$ sub-matrix and $\mathbf\Omega_j^c$ is the $(q-1)M\times (q-1)$ sub-matrix of $\mathbf\Omega$.
When $\mathbf A_{i,j}\in\mathbb F_q$ for all $i$ and $j$, $\mathbf\Omega$ coincides with the parity check matrix for the extended binary representation.
According to Lemma~\ref{lemma:Property_fw}, we also have the following properties of $\mathbf\Omega$.
\begin{lemma}\label{prop:Extended_parity_matrix}
    1) For all the non-zero $\mathbf A_{i,j},i\in\{1,2,\ldots,M\},j\in\{1,2,\ldots,N\}$, the column and row weights of the corresponding $\mathbf \Omega_{i,j}$ are all 1, \ie $\mathbf \Omega_{i,j}$ is a $(q-1)\times (q-1)$ permutation matrix.
    2) Let $\mathbf\Lambda_p(i,0),i=1,2,\ldots,M$ be the $i^{th}$ row vector of $\mathbf\Lambda_p$. Then the row weights of $\mathbf \Omega_i^T$ are the same and equal to the weight of $\mathbf\Lambda_p(i,0)$. Let $\mathbf\Lambda_p(0,j),j=1,2,\ldots,N$ be the $j^{th}$ column vector of $\mathbf\Lambda_p$. Then the column weights of $\mathbf\Omega_j^c$
    are equal to the weight of $\mathbf\Lambda_p(0,j)$.
    Degree distributions of $\mathbf\Omega$ are the same as those of $\mathbf\Lambda_p$.
\end{lemma}

Note that, If $\bar{\mathbf H}$ has the following form
\begin{equation}\label{eq:Proposed_parity_matrix}
      \bar{\mathbf H}=(\mathbf\Lambda_p\otimes\mathbf I_{p\times p})\cdot\mathrm{diag}(\mathbf\Psi_1,\mathbf\Psi_2,\ldots,\mathbf\Psi_N),
\end{equation}
where $\mathbf\Psi_j,j=1,2,\ldots,N$ are $p\times p$ full-rank matrices,
then the parity check matrix $\mathbf\Omega$ associated with Eq.~\eqref{eq:Proposed_parity_matrix} has the following form
\begin{equation}\label{eq:EPR_parity_matrix}
  \mathbf\Omega=(\mathbf\Lambda_p\otimes\mathbf{I}_{(q-1)\times(q-1)})\cdot\mathrm{diag}(f_{\omega}(\mathbf\Phi,\mathbf\Psi_1),\ldots,f_{\omega}(\mathbf\Phi,\mathbf\Psi_N)).
\end{equation}
According to the first item in Lemma~\ref{prop:Extended_parity_matrix}, $\mathbf\Omega$ in Eq.~\eqref{eq:EPR_parity_matrix} is composed of $q-1$ disjoint $\mathbf\Lambda_p$s. As a result, codes defined with these $\mathbf\Omega$s including the extended binary representation will cause performance loss.
Actually, if each $\mathbf\Psi_j$ is the matrix representation of a non-binary symbol in $\mathbb F_q$, Eq.~\eqref{eq:Proposed_parity_matrix} defines the equivalent binary parity check matrix of the non-binary column-scaled LDPC (CS-LDPC) code\cite{Zhao2013CsLDPC} (including the non-binary QC-LDPC codes and the finite geometry non-binary LDPC codes).

\subsection{General Framework for Constructing and Optimizing the EPR-LDPC Codes} 
\label{sub:general_framework_for_constructing_and_optimizing_the_epr_ldpc_codes}
In this subsection, we present a general framework for the exhaustive search of $\mathbf\Omega^e$.
Since the non-zero bits in $\mathbf v^e_j$ represent different combinations of the bits in $\bar{\mathbf x}_j$,
the parity check relationships for $\mathbf v^e$ can be obtained by finding the parity check relationships for the corresponding combinations
and the desired $\mathbf\Omega^e$ can be constructed by searching among different combinations of the parity check relationships for $\mathbf v^e$.

Definition~\ref{def:EPR} may imply that we should search for $\mathbf\Omega^e$ based on a given $\mathbf\Psi^e$.
However, in order to guarantee enhanced decoding performance for $\mathbf\Omega^e$, we first determine the desired $\mathbf\Omega^e$ then we update $\mathbf\Psi^e$.
That is,
\begin{enumerate}
    \item Based on $\bar{\mathbf H}$ and $\mathbf\Phi$, we find and store some of the parity check relationships for $\mathbf v$.
    \item Using these parity check relationships, we construct different $\mathbf\Omega^e$s row by row such that the new row does not introduce cycles smaller than certain integer.
    \item We find the $\mathbf\Omega^e$ with the desired performance threshold among these $\mathbf\Omega^e$s.
    Then, we update $\mathbf\Psi^e$ and $\mathbf v^e$.
\end{enumerate}

By utilizing $f_\omega$, we can model the searching processes (Step~2 and Step~3) as choosing proper $\hat{\mathbf\Psi}^e_j$ for all $j\in\{1,2,\ldots,N$\}.
Then $\mathbf\Omega^e$ is obtained by replacing each $\mathbf A_{i,j}$ in $\bar{\mathbf H}$ with $\sum_{i_q=(i-1)(q-1)+1}^{i(q-1)} f_{\omega}(\mathbf\Psi^{e^{i_q}}_j,\mathbf A_{i,j})$ and each $f_{\omega}(\mathbf\Psi^{e^{i_q}}_j,\mathbf A_{i,j})$ corresponds to a row in $\mathbf\Omega^e_{i,j}$ (some of $\mathbf\Psi^{e^{i_q}}_j$s could be zero matrices).

Note that $\mathbf A_{i,j}$s are not necessarily to be the matrix representations of the non-binary symbols of $\mathbb F_q$.
Moreover, we refer to $\mathbf\Psi^e_j(0,2^{i-1})\neq \mathbf 0,\forall i\in\{1,2,\ldots,p\}$ as the \emph{trivial} case for the EPR-LDPC code.
If each $\mathbf A_{i,j}$ is replaced by $\sum_{i_q=(i-1)(q-1)+1}^{i(q-1)} f_{\omega}(\mathbf\Psi^{e^{i_q}}_j,\mathbf A_{i,j})=f_{\omega}(\mathbf\Phi,\mathbf A_{i,j})$, the resulting $\mathbf\Omega^e$ coincides with the matrix $\mathbf\Omega$.
If $\mathbf\Psi^e_j\preccurlyeq\mathbf\Phi,j=1,2,\ldots,N$ and $\mathbf\Omega^e=\mathbf\Omega$, the resulting codes is a class of punctured EPR-LDPC codes for parity check matrix $\mathbf\Omega$ (the decoding messages for the punctured bits are set to be 0s in the soft decision decoders\cite{Jeongseok2004Punctured}).
It is easy to verify that codes in \cite{SV11ExtendedErasure} form a trivial case of this punctured EPR-LDPC code for non-binary LDPC code.

\subsection{Bit-level Cycles in $\mathbf\Omega$} 
\label{sub:bit_level_cycles_in_Omega}
In the previous subsection, the matrix map $f_\omega$ is introduced to give the definition of the EPR-LDPC code and formulate the exhaustive searching of $\mathbf\Omega^e$.
In this subsection, we investigate the relations between the symbol-level cycles in $\mathbf\Lambda_p$ and the bit-level cycles in $\mathbf\Omega$ based on the properties of $f_\omega$.
In general, we assume that $\mathbf\Lambda_p$ is of girth $g_h$. $\mathbf\Lambda_p$ is cycle-free if $g_h=0$.
Before the detailed demonstrations, we first give the definition for the matrix cycle.
\begin{definition}[matrix cycle]
  Given a binary parity check matrix $\bar{\mathbf H}$. Let $\mathbf\Lambda_p$ be its mother matrix.
  A \emph{matrix cycle} of length-$g$ in $\bar{\mathbf H}$ exists iff its corresponding positions in $\mathbf\Lambda_p$ form a symbol-level cycle of length-$g$.
\end{definition}

\begin{lemma}
  If the girth of the mother matrix $\mathbf\Lambda_p$ is $g_h>0$, then the girth of the associated parity check matrix $\mathbf\Omega$ for the $p$-reducible code is $g_s\geqslant g_h$, which is caused by the length-$g_s$ symbol-level cycle in $\mathbf \Lambda_p$. If $g_h=0$, $g_s=0$.
\end{lemma}
\begin{proof}
  Since $\mathbf\Omega_{i,j}$ is a $(q-1)\times (q-1)$ permutation matrix and cycle-free (due to the first item in Lemma~\ref{prop:Extended_parity_matrix}),
  if $\mathbf\Lambda_p$ contains no cycle, $\mathbf\Omega$ has no bit-level cycle.
  Moreover, a cycle in $\mathbf\Lambda_p$ can only cause matrix cycle in $\mathbf\Omega$ with the same length. The matrix cycle only contains bit-level cycles with the same length.
  And because $\mathbf\Omega_{i,j}$ is not equal to $\mathbf I_{(q-1)\times (q-1)}$,
  a matrix cycle can not always cause the bit-level cycles with the same length.
  Thus, the girth of the binary parity check matrix $\mathbf\Omega$ is not less than the girth of its mother matrix $\mathbf\Lambda_p$.
\end{proof}

For the non-binary parity check $\mathbf H$,
if $\mathbf H$ satisfies the cycle-free condition,
its associated $\mathbf\Omega$ will also satisfy the cycle-free condition.
A bit-level cycle in $\mathbf\Omega$ is caused by the symbol-level cycle of the same length in $\mathbf H$.
Moreover, when $\mathbf H$ is constructed with cycles, investigations indicate that, among the cycles, the length-4 cycles contribute the most to the performance degradation.
Next, we show that a length-4 symbol-level cycle in $\mathbf H$ will not always result in length-4 bit-level cycles in $\mathbf\Omega$.

\begin{theorem}\label{thm:L4_prb}
  Let the non-zero matrix labels be uniformly taken from $\mathbb F^*_q$.
  And the probability that a length-4 symbol-level cycle in the non-binary parity check matrix $\mathbf H$ will result in length-4 bit-level cycles in $\mathbf\Omega$ based on $f_\omega$ is denoted by $p_4$.
  Then \[
    p_4= \frac{1}{q-1}
  \] for $q=2^p\geqslant 4$.
\end{theorem}
\begin{proof}
  Since the length-4 bit-level cycles are only caused by the length-4 symbol level cycle, we only consider the bit-level cycles within a symbol-level cycle. Let $(i_1,j_1)$, $(i_1,j_2)$, $(i_2,j_1)$, $(i_2,j_2)$ be the four coordinates of four entries that represent a length-4 symbol level cycle  in $\mathbf H$.
  Let
  \begin{equation}
  \nonumber
  \left(\begin{array}{cccc}
  \mathbf\Omega_{i_1,j_1} & \mathbf\Omega_{i_1,j_2}\\
  \mathbf\Omega_{i_2,j_1} & \mathbf\Omega_{i_2,j_2} \\
  \end{array}\right)
  \end{equation}
  be the matrix cycle corresponding to a length-4 symbol-level cycle.
  Let $\alpha_1,\beta_1,\alpha_2,\beta_2\in\{1,2,\ldots,q-1\}$ respectively represent the column numbers of non-zero entries in $\mathbf\Omega_{i_1,j_1}$, $\mathbf\Omega_{i_1,j_2}$, $\mathbf\Omega_{i_2,j_1}$ and $\mathbf\Omega_{i_2,j_2}$ with $\alpha_1, \beta_1$ in the same row and $\alpha_2, \beta_2$ in the same row.
  Let
    $\mathscr S_1=\{(\alpha_1,\beta_1),\alpha_1,\beta_1\in\{1,2,\ldots,q-1\}\}$ and $
        \mathscr S_2=\{(\alpha_2,\beta_2),\alpha_2,\beta_2\in\{1,2,\ldots,q-1\}\}$
  be the two-tuple sets containing all the different rows in $(\mathbf\Omega_{i_1,j_1},\mathbf\Omega_{i_1,j_2})$
  and $(\mathbf\Omega_{i_2,j_1},\mathbf\Omega_{i_2,j_2})$, respectively.
  Then, $|\mathscr S_1|=|\mathscr S_2|=q-1$.
  Let $\mathscr S$ be the set containing all the rows that could be involved in the length-4 matrix cycles.
  Let $\mathscr S=\{(\alpha,\beta),\alpha,\beta=1,2,\ldots,q-1\}$ and $|\mathscr S|=(q-1)^2$
  with $\mathscr S_1,\mathscr S_2\subset \mathscr S$.
  The length-4 bit-level cycle exist iff
  $\mathbf{Pr(\mathscr S_1\cap\mathscr S_2\neq \emptyset)}=1-\mathbf{Pr(\mathscr S_1\cap\mathscr S_2=\emptyset)}.$
  We can calculate the probability of $\mathscr S_1\cap\mathscr S_2=\emptyset$ by counting the number of choices of $\mathscr S_1$ and $\mathscr S_2$ over $\mathscr S$.
  Since there are $q-1$ different non-zero $\mathbf\Omega_{i,j}$s, different $\mathbf\Omega_{i,j}$s have different row numbers of the same row-vectors and no two different $\mathscr S_i$s have common elements, different $\mathscr S_i$s divide $\mathscr S$ into $q-1$ disjoint subsets.
  And because each $\mathscr S_i$ is uniformly chosen, then for a $\mathscr S_1$, there exist $(q-2)$ $\mathscr S_2$s that do not form cycles.
  As a result, $\mathbf{Pr(\mathscr S_1\cap\mathscr S_2=\emptyset)}=\frac{(q-1)(q-2)}{(q-1)^2}.$
\end{proof}

From Theorem~\ref{thm:L4_prb}, we know that, as the field size increases, the probability that a length-4 symbol-level cycle in $\mathbf H$ results in length-4 bit-level cycles in $\mathbf\Omega$ will be significantly reduced.
\begin{corollary}\label{corol:L4_prob}
  For any $p$-reducible code, let the matrix labels be chosen uniformly over a set $\{\mathbf B_g,g=1,2,\ldots,Q\}$. If there exist an integer $P\leqslant Q$ such that $\mathrm{rank}(f_\omega(\mathbf\Phi,\mathbf B_{g_i})+f_\omega(\mathbf\Phi,\mathbf B_{g_j}))=q-1$ for all $i\neq j, i,j\in\{1,2,\ldots,P\}$, then the probability that a length-4 symbol-level cycle in $\mathbf\Lambda_p$ will result in length-4 bit-level cycles in $\mathbf\Omega$, i.e., $p_4^\prime$, satisfies
\begin{equation}\label{eq:L4_p_reducible}
  \frac{1}{q-1}\leqslant p_4^\prime\leqslant \frac{1+(Q-P)^2}{P+(Q-P)^2}
\end{equation}
and $P\leqslant q-1$ for $q=2^p\geqslant 4$. When $P=1$, $p_4^\prime=1$.
\end{corollary}
\begin{proof}
  The $P$ matrix labels result in at most $q-1$ disjoint subsets of $\mathscr S$ then $P\leqslant q-1$. The proof for the above inequality which results from the different values of $Q-P$ is similar to the proof of Theorem~\ref{thm:L4_prb}.
\end{proof}

According to Corollary~\ref{corol:L4_prob}, $p_4^\prime$ can be minimized by enlarging $q$ and minimizing $Q-P$.
Consider a short length matrix cycle of length-$g_c,g_c>4$. Based on the proof of Theorem~\ref{thm:L4_prb}, we suppose that the probability of the existence of the corresponding bit-level cycles is small too and relates to both $q$ and $g_c$.
We also have the following observation for the short length cycles with lengths larger than 4.

\begin{observation}\label{obs:prob_larger_bit_cycles}
\mbox{}
	\begin{enumerate}
		\item For a $p$-reducible code in Corollary~\ref{corol:L4_prob}, the probability that a symbol-level cycle of length-$g_c$ in $\mathbf H$ will cause corresponding bit-level cycles in $\mathbf \Omega$ is bigger than $\frac{1}{q-1}$.
		\item This probability increases as the length of the symbol-level cycle increases and decreases as $q=2^p$ increases.
	\end{enumerate}
\end{observation}

\subsection{Design of EPR-LDPC Codes according to $\mathbf\Omega$} 
\label{sub:design_of_epr_ldpc_codes_according_to_omega}
In this subsection, we show how to efficiently find the parity check matrix $\mathbf\Omega^e$ with certain girth.
The resulting $\mathbf\Omega^e$ can be also seen as a significant generalization of $\mathbf\Omega$ for both binary and non-binary LDPC codes.
Compared to the $\mathbf\Omega$, $\mathbf\Omega^e$ will have less bit-level cycles.
In addition, superior to $\mathbf\Omega$ in Eq.~\eqref{eq:EPR_parity_matrix}, $\mathbf\Omega^e$ will not be composed of disjoint sub-matrices.
And for any girth-optimized $p$-reducible code, $\mathbf\Omega$ can be constructed to further improve the decoding performance.

\begin{figure}[!hbtp]
\begin{center}
\includegraphics[width=.45\textwidth]{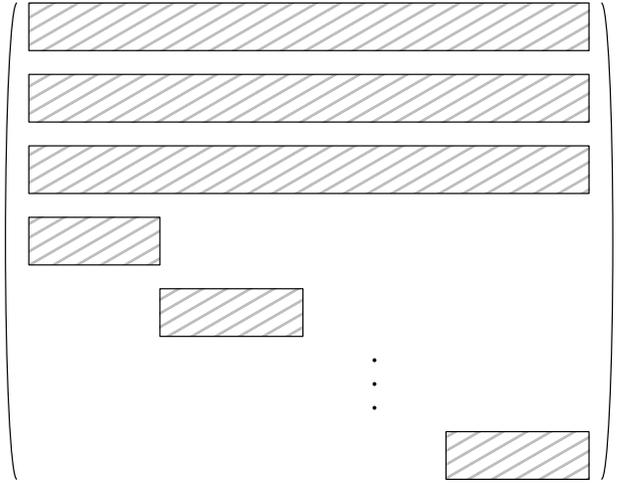}
\end{center}
\vspace*{-.3cm}
\caption{The structure of matrix $\mathbf\Omega^e$.}
\label{fig:omega_e}\vspace*{-.3cm}
\end{figure}

First, according to Observation~\ref{obs:prob_larger_bit_cycles}, if $q$ is large enough, $\mathbf\Omega$ can be constructed with the short length bit-level cycles being largely avoided in many cases. In these cases, we can always further avoid some short length bit-level cycles by simply changing the associated matrix labels $\mathbf A_{i,j}$s in $\bar{\mathbf H}$ and searching for the proper $\mathbf\Omega_{i,j}$s that do not form cycles.
However, for some $p$-reducible codes, e.g., the codes in Corollary~\ref{corol:L4_prob} with $P=1$ and the non-binary CS-LDPC codes etc., the short length symbol-level cycles will always cause corresponding short length bit-level cycles in $\mathbf\Omega$s.
In these cases, if $\mathbf\Lambda_p$ is constructed with some length-4 symbol-level cycles, we can still avoid few length-4 bit-level cycles in $\mathbf\Omega$.
That is, if the weight of one row is 3, we can add the smaller weight row to the larger weight row such that the resulting row is not of weight-1.
However, this row addition operation can only handle limited number of length-4 bit-level cycles.
As to the short-length cycles with lengths larger than 4, it is hard to avoid them without altering the structure of $\mathbf\Omega$.
In addition, the $\mathbf\Omega$ in Eq.~\eqref{eq:EPR_parity_matrix} is composed of disjoint sub-matrices which will cause performance loss.
So, we need to find another representation of $\bar{\mathbf H}$ which will generally result in good performance.
Actually, since $\mathbf\Omega$ has already avoided a large number of corresponding short length bit-level cycles in many cases and some cycles in $\mathbf\Omega$ can be carefully handled, we could find the desired representation, i.e., $\mathbf\Omega^e$, more efficiently from $\mathbf\Omega$ instead of searching among numerous parity check combinations.
We give the details below.
\begin{description}
	\item[Step~1]: Let $q=2^p$ and $p>1$. We construct a binary matrix set $\{\mathbf B^\prime_1,\mathbf B^\prime_2,\mathbf B^\prime_3,\ldots\}$ with each $\mathbf B^\prime_i$ being a cycle-free $2\times (q-1)$ or $2\times 2(q-1)$ matrix.
	In addition, $\mathbf B^\prime_i\cdot\mathbf v_j=\mathbf 0, \forall i,j$ or $(\mathbf B^\prime_i(0,1),\ldots,\mathbf B^\prime_i(0,q-1))\cdot\mathbf v_j=\mathbf 0$ and $(\mathbf B^\prime_i(0,q),\ldots,\mathbf B^\prime_i(0,2q-1))\cdot\mathbf v_j=\mathbf 0,\forall i,j$.
	\item[Step~2]: Given a parity check matrix $\bar{\mathbf{H}}$ in Definition~\ref{def:p_reducible} with mother matrix $\mathbf\Lambda_p$. We construct $\mathbf\Omega$ by using $f_\omega$. If $\mathbf\Lambda_p$ is constructed with length-4 cycles, we use the row addition operation to eliminate some length-4 bit-level cycles.
	\item[Step~3]: Let $g_s$ be an even number. For the matrix cycles with length less than $g_s$ that results in bit-level cycles in $\mathbf\Omega$, we set the rows across the matrix cycles to be zero vectors and rearrange these zero rows to the lower part of the resulting matrix. Then, as illustrated in Fig~\ref{fig:omega_e}, we place the $\mathbf B^\prime_i$s that will not cause bit-level cycles with length less than $g_s$ one by one within these zero rows (at the non-overlapped column-positions).
	The resulting matrix is denoted by $\mathbf\Omega^e$.
\end{description}

Note that, given a practical $p$-reducible LDPC code, the row addition operation in Step~2 can be omitted as the length-4 cycles in $\mathbf\Lambda_p$ are in general eliminated.
Then, we can only use the row replacing operation in Step~3 to handle the bit-level cycles.
Moreover, the codes defined with $\mathbf\Omega^e$ and $\mathbf\Omega$ tend to have larger code lengths and code spaces than their associated binary LDPC codes.
How to obtain the correct $\bar{\mathbf x}$ and avoid the undetected errors for $\mathbf v^e$ apart from enlarging the girth of $\mathbf\Omega^e$ will be addressed in the next section.


\section{Bit-Level Decoders for the EPR-LDPC codes} 
\label{sec:Class_Faster_DecodersF}
\subsection{General Sum-Product Decoding} 
\label{sub:general_sum_product_decoding}
Let $\mathcal{C}$ be the non-binary LDPC code. Let $\bar{\mathcal C}$ be the equivalent binary LDPC code.
Let $\mathcal{C}^e$ be the EPR-LDPC code.
Let $\mathcal{C}^e_j$ be the binary LDPC code generated by $\mathbf\Psi^e_j$.
Obviously, there exists the following isomorphism
\begin{equation}\label{eq:Isomorphism}
  \mathcal{C}\cong\bar{\mathcal C}\cong\mathcal{C}^e\cap(\mathcal{C}^e_1\times\mathcal{C}^e_2\times\cdots\times\mathcal{C}^e_N).
\end{equation}
The above equation implies that to obtain an enhanced decoding performance of $\bar{\mathcal C}$, we should decode $\mathbf v^e$ by utilizing the parity check relationships for $\mathcal{C}^e$ and $\mathcal{C}^e_1\times\mathcal{C}^e_2\times\cdots\times\mathcal{C}^e_N$ simultaneously.
If $\mathbf\Psi^e_j=\mathbf\Phi,j=1,2,\ldots,N$, Eq.~\ref{eq:Isomorphism} represents the isomorphism between (punctured) extended binary representation and its non-binary counterpart\cite{Savin09Erasure,SV11ExtendedErasure}.
The decoding applications of the extended binary representation over general channel models are given in\cite{Savin2012Fourier}.
For arbitrary $p$-reducible code, the above isomorphism also exists.
To obtain the best decoding results of $\bar{\mathcal C}$, each $\mathrm{wt}(\mathbf\Psi^e_j)$ should be large enough (more parity check bits will be involved).
On the other hand, large $\mathrm{wt}(\mathbf\Psi^e_j)$ will results in higher decoding complexity in general.
Thus, there exist a trade-off between the choices of $\mathbf\Psi^e_j,j=1,2,\ldots,N$ and the decoding performance.
Then to have the optimized $\mathbf\Psi^e_j$s, we have to maximize each $\mathrm{wt}(\mathbf\Psi^e_j)$ while minimize the probability of the existence of the short length bit-level cycles.
Next, we first show how to obtain $\bar{\mathcal C}$.
Then we optimize the extended generator matrices set.

Assume that $\bar{\mathbf x} = (\bar{\mathbf x}_1^T,\bar{\mathbf x}_2^T,\ldots,\bar{\mathbf x}_N^T)^T$ is transmitted over binary input channels. Let $\bar{\mathbf y} = (\bar{\mathbf y}_1^T,\bar{\mathbf y}_2^T,\ldots,\bar{\mathbf y}_N^T)^T$ be the received sequence.
The proposed decoder is a class of binary decoders which is implemented to make decisions both on $\bar{\mathbf x}$ and $\mathbf v^e$. In the following, we first give the general sum-product decoding procedure for the proposed decoders. Then we develop two variants for different channel models.

\begin{figure}[!hbtp]
\begin{center}
\includegraphics[width=.45\textwidth]{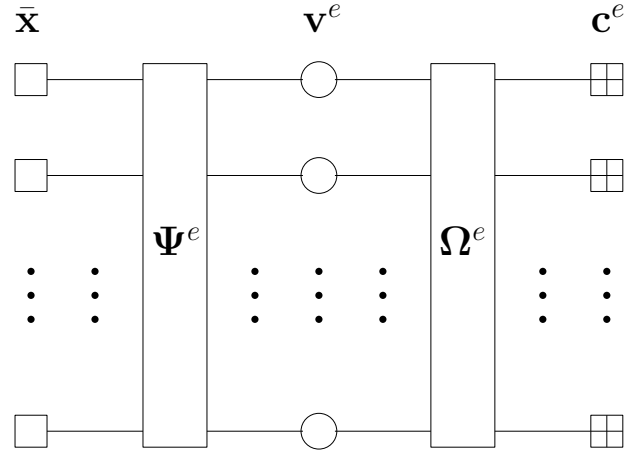}
\end{center}
\vspace*{-.3cm}
\caption{General decoding procedure for the proposed binary decoders.}
\label{fig:General_decoding}\vspace*{-.3cm}
\end{figure}

As shown in Fig~\ref{fig:General_decoding}, the bits in $\bar{\mathbf x}$ are represented as \emph{bit nodes}. The bits in $\mathbf v^e$ are represented as \emph{extended bit nodes}. Every rows of $\mathbf\Omega^e$ are represented as \emph{constraint nodes} in $\mathbf c^e$. Then the general decoding procedure is described as follows.
\begin{description}
  \item[Step~1]: Let $\mu^{(0)}_x$ be the message from channels. The message $\mu^{(0)}_v$ for $\mathbf v^e$ is calculated from $\mu^{(0)}_x$ according to $\mathbf\Psi^e$.
  \item[Step~2]: The message $\omega_c^{(l)}$ in constraint nodes is calculated from $\mu^{(l)}_v$ according to parity check matrix $\mathbf\Omega^e$ for the EPR-LDPC code.
  \item[Step~3]: The message $\mu^{(l+1)}_v$ in the extended bit node is updated from $\omega_c^{(l)}$ according to $\mathbf\Omega^e$.
  \item[Step~4]: The message $\mu^{(l+1)}_v$ is further tailored based on the generator matrices set $\mathbf\Psi^e$.
  \item[Step~5]: For iteration-$h$, if the hard decision of $\mathbf v^e$ is $\hat{\mathbf v}^e$ which satisfies $\mathbf\Omega^e\hat{\mathbf v}^e=\mathbf 0$, then we obtain $\bar{\mathbf x}$ from $\mathbf v^e$ according to $\mathbf\Psi^e$.
\end{description}

Since we use a binary decoding process and the zero columns in each $\mathbf\Psi^e_j$ and $\mathbf\Omega^e$ can be removed, the computational complexity for the check-vector-sum operation relies linearly on the number of the non-zero columns in $\mathbf\Omega^e_i,i=1,2,\ldots,M$.
The computational complexity of tailoring $\mu^{(l+1)}_v$ relies linearly on the non-zero columns in $\mathbf\Psi^e_j,j=1,2,\ldots,N$.
Let the maximum number of the non-zero columns in each $\mathbf\Omega^e_i$ be $\phi_e\leqslant q-1$ and
the maximum number of the non-zero columns in each $\mathbf\Psi^e_j$ be $\psi_e\leqslant q-1$.
Then the computational complexity is dominated by $O(m_s=\max\{\phi_e,\psi_e\})$.

In the general decoding procedure, when the decoding of $\mathbf v^e$ over $\mathbf\Omega^e$ is accomplished, we have to get every $\bar{\mathbf x}_j$ from $\mathbf v^e_j$.
To guarantee $\bar{\mathbf x}_j$ being successfully resolved from $\mathbf v^e_j$,
we also provide the following conditions for the extended generator matrices.
\begin{theorem}
\label{them:resolvable}
Consider the $p$-reducible codes.
For all $j\in\{1,2,\ldots,N\}$ and $q=2^p\geqslant 4$,
\begin{enumerate}
  \item if $\mathrm{wt}(\mathbf\Psi^e_j)>\frac{q}{2}-1$, every bits in $\bar{\mathbf x}_j$ can be resolved from $\mathbf v^e_j$.
  \item If $\mathrm{wt}(\mathbf\Psi^e_j)=\frac{q}{2}-1$, $\bar{\mathbf x}_j$ can be resolved with probability of $1-\frac{q-1}{{{q-1}\choose{\frac{q}{2}-1}}}$.
\end{enumerate}
\end{theorem}
\begin{proof}
  Recall that $\mathbf\Phi$ is a $p\times(q-1)$ matrix and \[V=\{\mathbf 0, \mathbf\Phi(0,1),\mathbf\Phi(0,2),\ldots,\mathbf\Phi(0,q-1)\}\] is a vector space of dimension-$p$. Let \[V^e_j=\{\mathbf 0, \mathbf\Psi^e_j(0,1),\mathbf\Psi^e_j(0,2),\ldots,\mathbf\Psi^e_j(0,q-1)\}\] be the set that formed by the column vectors of $\mathbf\Psi^e_j$. Then $\mathrm{wt}(\mathbf\Psi^e_j)=|V^e_j|-1.$
  Let the set $V^\prime=\{\mathbf\Phi(0,1),\mathbf\Phi(0,2),\ldots,\mathbf\Phi(0,2^{p-1})\}$ be the set of all unit vectors. Then the non-zero vectors in $V$ and $V^e_j$ can be formulated by the additions of the vectors in $V^\prime$.

  If $|V^e_j|$ is larger than the size of the $(p-1)$-dimensional subspace of $V$, then $\mathrm{rank}(\mathbf\Psi^e_j)=p$. Every bits in $\bar{\mathbf x}_j$ can be resolved.
  The size of the $(p-1)$-dimensional subspace can be calculated by
  $\sum_{i=1}^{p-1}{{p-1}\choose{i}}+1=2^{p-1}$.
  Then if $\mathrm{wt}(\mathbf\Psi^e_j)>\sum_{i=1}^{p-1}{{p-1}\choose{i}}=2^{p-1}-1$, $\bar{\mathbf x}_j$ can be resolved from $\mathbf v^e_j$.
  If $\mathrm{wt}(\mathbf\Psi^e_j)=\sum_{i=1}^{p-1}{{p-1}\choose{i}}$, the rank of $\mathbf\Psi^e_j$ is either $p$ or $p-1$. Then the probability that $\bar{\mathbf x}_j$ can be resolved equals the probability that the non-zero vectors in $\mathbf\Psi^e_j$ do not form a $(p-1)$-dimensional subspace, which depends on the number of the $(p-1)$-dimensional subspaces.
  To calculate the number of the $(p-1)$-dimensional subspaces of the $V$, we first introduce the Gaussian binomial coefficient over finite field $\mathbb F_q$
  \[
    {n\choose{k}}_q=\frac{[n]_q!}{[k]_q![n-k]_q!},k\leqslant n,
  \]
  where $[n]_q!=[1]_q[2]_q\cdots[n]_q$ with
  \begin{eqnarray}\nonumber
  	[m]_q&=&\frac{1-q^m}{1-q}\\\nonumber
  	&=&\sum_{0\leqslant i<m} q^i=1+q+q^2+\cdots q^{m-1},1\leqslant m\leqslant n.
  \end{eqnarray}
  Then the number of the the $(p-1)$-dimensional subspaces over $\mathbb F_2$ is calculated by
    ${p\choose{p-1}}_2 = \sum_{i=1}^p2^{i-1}$.
  The probability that $\bar{\mathbf x}_j$ can be resolved when $\mathrm{wt}(\mathbf\Psi^e_j)=\sum_{i=1}^{p-1}{{p-1}\choose{i}}$ is \[1-{p\choose{p-1}}_2\big{/}{{q-1}\choose{\mathrm{wt}(\mathbf\Psi^e_j)}}.\]
\end{proof}

Note that $2^{p-1}=\sum_{i=1}^{p-1}{{p-1}\choose{i}}+1\geqslant p,\forall p\geqslant 2$.
In addition, if $\bar{\mathbf x}_j$ can be resolved, $\mathrm{wt}(\mathbf\Psi^e_j)$ is at least the size of a basis of a dimension-$p$ vector space over $\mathbb F_q$, i.e., $\mathrm{wt}(\mathbf\Psi^e_j)\geqslant \log_2q,j=1,2,\ldots,N$.
Thus, the least number of non-zero columns required for each $\mathbf\Psi^e_j$ is $p$.
However, if $\mathrm{wt}(\mathbf\Psi^e_j)=p,\forall j$, the desired parity check matrix $\mathbf\Omega^e$ may not exist. Theorem~\ref{them:resolvable} provides sufficient conditions for the successful decoding of $\bar{\mathbf x}$. Moreover, large $\mathrm{wt}(\mathbf\Psi^e_j)$ generally results in better decoding performance. For the punctured EPR-LDPC code decoded over parity check matrix $\mathbf\Omega$, if all the punctured bits are recovered we can use $\mathbf\Phi$ instead of $\mathbf\Psi^e_j$ to resolve $\bar{\mathbf x}_j$.

\subsection{Decoding over Different Channels} 
\label{sub:decoding_over_different_channels}
Next, we apply the general decoding procedure to the binary symmetric channel (BSC).
\begin{example}\label{exam:Hard}
  In this example, we present an extended iterative hard decision decoder for BSC.
  Let $\boxplus$ be the bit-wise addition of the vector space over $\mathbb F_2$.
  Then, for simplex code\cite{MacWilliams1978FEC}, we have $\mathbf v_j(j^\prime_1)+\mathbf v_j(j^\prime_2)+\cdots +\mathbf v_j(j^\prime_k)=\mathbf v_j(j^\prime_1\boxplus j^\prime_2\boxplus\cdots \boxplus j^\prime_k),j^\prime_i\in\{1,2,\ldots,q-1\}$\cite{Savin09Erasure}.
  As a result, the iterative decoding procedure is described as follows.
  \begin{description}
    \item[Step~1]: Let $\hat{\mathbf v}^e$ be the message for the extended bit nodes which is initialized by the value of $\mathbf\Psi^{eT}_j\bar{\mathbf y}_j,j=1,2,\ldots,N$ and $b$ be the thresholds to perform the bit-flipping.
    \item[Step~2]: If $\mathbf z=\mathbf\Omega^e\hat{\mathbf v}^e=\mathbf 0$ then $\mathbf v^e=\hat{\mathbf v}^e$.
    Else, $\mathbf s=\mathbf z^T\mathbf\Omega^e=(\mathbf s_j)_{1\times N}$ (here is the decimal multiplication).
    For all $\mathbf s_j(j^\prime)\geqslant b,j^\prime\in\{1,2,\ldots,q-1\}$, if there exist $\mathbf v^e_j(j^\prime_1)+\mathbf v^e_j(j^\prime_2)+\cdots +\mathbf v^e_j(j^\prime_k)\neq\mathbf v^e_j(j^\prime)$
    where $j^\prime_i\in\{1,2,\ldots,q-1\}$ such that $\mathbf\Psi^e_j(0,j^\prime_i)\neq \mathbf 0$ and $j^\prime=j^\prime_1\boxplus j^\prime_2\boxplus\cdots \boxplus j^\prime_k$,
    then $\hat{\mathbf v}^e_j(j^\prime)=1+\hat{\mathbf v}^e_j(j^\prime)$.
    \item[Step~3]: Stop the procedure when $\mathbf\Omega^e\hat{\mathbf v}^e=\mathbf 0$ or the maximum iteration number is reached. Then for the trivial case, $\bar{\mathbf x}_j=({\mathbf v}_j^e(1),{\mathbf v}_j^e(2),\dots,{\mathbf v}_j^e(2^{p-1}))^T$.
  \end{description}
\end{example}

Below, we show how to apply the BP algorithm into the decoding of the EPR-LDPC code over binary input Gaussian channel.
\begin{example}\label{exam:Gaussian}
  We give a hybrid parallel decoder for the EPR-LDPC codes by using the BP decoder and the hard decision decoder in example~\ref{exam:Hard}. The BP decoder and hard decision decoder exchange decoding messages iteratively.
  We consider one decoding round is finished iff these two decoders have exchanged information once. A $(\mu,\nu)$ decoding round is a decoding round within which the BP decoder has performed $\mu$ times decoding iterations and the hard decision decoder has performed $\nu$ times decoding iterations.
  Different from example~\ref{exam:Hard}, we choose to transmit $\mathbf v^e$ instead of $\bar{\mathbf x}$ as in the general decoding procedure.
  Assume BPSK is utilized. Let $\mathbf y^e$ be the received sequence. Then the decoding process is described below.
  \begin{description}
    \item[Step~1]: Initialize the message for the $v^{th}$ extended bit node by $\mu_{v,c}^{(0)}=\frac{2}{\sigma^2}\mathbf y^e(v)$ and the message for the $c^{th}$ constraint node by $\omega_{c,v}^{(0)}=0$.
    \item[Step~2]: $\omega_{c,v}^{(l)}=-2\tanh^{-1}\left(\prod_{i^{\prime\prime}\in\mathcal{N}_c\backslash\{v\}}\tanh\left(\frac{-\mu_{i^{\prime\prime},c}^{(l-1)}}{2}\right)\right)$, where $\mathcal N_c$ is the extended bit nodes set connected to the $c^{th}$ constraint node.
    \item[Step~3]: $\mu_{v,c}^{(l)}=\frac{2}{\sigma^2}\mathbf y^e(v)+\sum_{j^{\prime\prime}\in\mathcal M_{v}\backslash\{c\}}\omega_{j^{\prime\prime},v}^{(l)}$, where $\mathcal M_{v}$ is the constraint nodes set connected to the $v^{th}$ extended bit node.
    \item[Step~4]: For iteration-$\mu$ in a $(\mu,\nu)$ decoding round, let the hard decision be $\hat{\mathbf v}^e$. We apply the hard decision decoder in example~\ref{exam:Hard} for $\nu$ times. If $\hat{\mathbf v}^e(v)=1$, $\mu_{v,c}^{(l)}=|\mu_{v,c}^{(l)}|$, else $\mu_{v,c}^{(l)}=-|\mu_{v,c}^{(l)}|$. Then, go to step 2.
    \item[Step~5]: For iteration-$h\mu\nu$, if the hard decision $\hat{\mathbf v}^e$ satisfies $\mathbf\Omega^e\hat{\mathbf v}^e=\mathbf 0$, then $\mathbf v^e=\hat{\mathbf v}^e$ and
    $\bar{\mathbf x}_j=({\mathbf v}_j^e(1),{\mathbf v}_j^e(2),\dots,{\mathbf v}_j^e(2^{p-1}))^T$ for the trivial case.
  \end{description}

If $\bar{\mathbf x}$ is transmitted in example~\ref{exam:Gaussian}, the initialization of the messages for the extended bit nodes can be performed as follows.
First, let $\mathbf v^e(v)=\sum_{i^{\prime}\in\mathcal S_v}\bar{\mathbf x}(i^{\prime})$,
where $\mathcal S_v$ is the set containing all the bit nodes connected to the $v^{th}$ extended bit node.
Let $\mathrm{hard}(\cdot)$ be the hard decision function.
And $\mathrm{hard}(\mathbf v^e(v))=\sum_{i^{\prime}\in\mathcal S_v}\mathrm{hard}(\bar{\mathbf x}(i^{\prime}))$.
Then the initialized messages for the extended bit node $\mathbf v^e(v)$ are
\begin{equation} \nonumber
  \left\{
  \begin{aligned}
    -\frac{2}{\sigma^2}\min_{i^{\prime}\in\mathcal S_v}(|\bar{\mathbf y}(i^{\prime})|), \text{     if } \mathrm{hard}(\mathbf v^e(v))=1, \\
    \frac{2}{\sigma^2}\min_{i^{\prime}\in\mathcal S_v}(|\bar{\mathbf y}(i^{\prime})|), \text{     if } \mathrm{hard}(\mathbf v^e(v))=0.
  \end{aligned}
  \right.
\end{equation}
The decoding procedure is the same as in example~\ref{exam:Gaussian}. Each bit in $\mathbf v^e$ is then transmitted over a copied Gaussian channel. To evaluate the performance of an EPR-LDPC code ensemble over copied Gaussian channels, we use the Monte-Carlo experiments for infinite LDPC codes introduced in \cite{Davey98montecarlo}.
That is, by simulating an infinite long EPR-LDPC code from the code ensemble, we evaluate the performance in terms of the minimum signal to noise ratio (MSNR), i.e., $T_b$, for which the average syndrome bit entropy reaches certain value after a number of decoding iterations.
\end{example}

\subsection{Code Optimization for The Binary Decoders} 
\label{sub:Profiles_optimization}
For the mother matrix $\mathbf\Lambda_p$, how to optimize the matrix labels for the non-binary decoders have been studied in \cite{Poulliat08BinImage,Davey98montecarlo}.
The authors in\cite{Poulliat08BinImage,Davey98montecarlo} propose several optimization methods based on the equivalent binary LDPC codes.
The degree distributions for the resulting $\bar{\mathbf H}$ can be efficiently calculated according to \cite{Yang12Erasure}.
As to the EPR-LDPC code, we can optimize the matrix labels according to Corollary~\ref{corol:L4_prob}.
For different associated $p$-reducible codes, the optimized matrix labels for the same $p$ will also be very different because different $p$-reducible codes use matrix labels set with different structures.
Moreover, to obtain enhanced decoding performance of EPR-LDPC codes, just optimizing the matrix labels is not enough.
We have to carefully construct the parity check matrix $\mathbf\Omega^e$ and the extended generator matrices set $\mathbf\Psi^e$ too.
In the following, we present a simple algorithm based on the proposed framework to achieve this goal.
That is, after we have the matrix labels optimized, we guarantee that each $\mathrm{wt}(\mathbf\Psi^e_j)$ is large enough and optimize the girth and degree distributions of the $\mathbf\Omega^e$.
$\mathbf\Lambda_p$ is assumed to be constructed by the modified progressive-edge-growth (PEG) algorithm.
One also can construct $\mathbf\Lambda_p$ by other random methods or with specific structures.
We can either fix the code length and change $\mathbf\Lambda_p$ or fix the $\mathbf\Lambda_p$ and change the code length for different purposes.
Details are given as follows.
\begin{description}
  \item[Step~1]: The binary parity check matrix $\bar{\mathbf H}$ is obtained by filling $\mathbf\Lambda_p$ with the optimized matrix labels of size $p\times p$ according to Corollary~\ref{corol:L4_prob}.
  Let $\psi>\frac{q}{2}-1$ and $\phi>0$ be two non-zero integers.
  Let $T_b$ be the MSNR in dB.
  We search for $\mathbf\Omega^e$ with the generator matrix set $\mathbf\Psi^e$ and the associated $\hat{\mathbf\Psi}^e_j,j=1,2,\ldots,N$ satisfying $\mathrm{wt}(\mathbf\Psi^e_j)\geqslant\psi,j=1,2,\ldots,N$,
  $\mathrm{wt}(\hat{\mathbf\Psi}^{e^i}_j)\geqslant\phi,i=1,2,\ldots,M$.
  And the MSNR for the resulting degree distributions does not exceed $T_b$ (for short block length codes, we drop the MSNR examinations).
  \item[Step~2]: For an even number $g_s$, we check if the girth of $\mathbf\Omega^e$ is not smaller than $g_s$.
  If the girth of $\mathbf\Omega^e$ is smaller than $g_s$, then $p=p+1$ and go to step 1.
  \item[Step~3]: When $p$ is large enough, we may change some $\mathbf\Psi^{e^{i_q}}_j$s which generate the rows across the associated matrix cycles to further eliminate some bit-level cycles. Then we update $\hat{\mathbf\Psi}^e_j,j=1,2,\ldots,N$ and its associated $\mathbf\Psi^e$.
\end{description}

\section{Simulation} 
\label{sec:Simulation}
\subsection{Different binary forms of a non-binary LDPC code} 
\label{sub:Different_binary_forms}
In this subsection, we present the simulation results for different representations of a non-binary LDPC code under different decoders. No undetectable error is observed in our simulations.
We consider the code over $\mathbb F_{8}$ of rate $R=0.5311$ with length-12000 bits.
To have a fair comparison, the code we have found has similar MSNR for different representations.
Let $\mathbf v^e=\mathbf v$ and $\mathbf\Omega^e_i\neq\mathbf\Omega_i$ for some $i$, i.e. the block lengths for $\mathbf\Omega^e$ and $\mathbf\Omega$ are the same.
We search for $\mathbf\Omega^e$ by using the method in Section~\ref{sub:Profiles_optimization}.
Degree distributions and MSNRs for ${\mathbf H}$, $\bar{\mathbf H}$, $\mathbf\Omega^e$ and $\mathbf\Omega$ are displayed in table~\ref{tab:DegreeDist_1} and table~\ref{tab:DegreeDist_2}.
The code defined with $\mathbf\Omega^e$ is $R=0.5355$.
The MSNR for $\mathbf\Omega^e$ is $E_b/N_0=0.73$dB.
The MSNR for $\mathbf\Omega$ is $E_b/N_0=0.68$dB, while the capacity limit is $E_b/N_0=0.30$dB.
The comparison is shown in Fig~\ref{fig:Decoders}, where
HEPR (hard decision decoder for the EPR-LDPC code) is the extended hard decision decoder for $\mathbf\Omega^e$,
SEPR (soft decision decoder for the EPR-LDPC code) is the hybrid parallel decoder for $\mathbf\Omega^e$,
QSPA is the $q$-ary sum-product decoder for $\mathbf H$,
SEB (soft decision decoder for the equivalent binary LDPC code ) is the binary BP decoder for $\bar{\mathbf H}$ and
SER (soft decision decoder for the extended binary representation) is the hybrid parallel decoder for $\mathbf\Omega$.
Due to the short length bit-level cycles in $\bar{\mathbf H}$,
SEB suffers from a performance loss of about $1$dB.
Decoding complexity per each check-sum for QSPA is $O(q^2)$.
In our simulation, SEPR achieves second place with much lower decoding complexity
while the performance gap to QSPA is within $0.2$dB.
Moreover, SEPR outperforms SER for the same block length.

Consider the non-binary LDPC code of rate half over $\mathbf F_{16}$ characterized by
\begin{eqnarray}\nonumber
	\lambda(x)&=&0.303x + 0.337x^2 + 0.04x^3 + 0.113x^4 + \\\nonumber
	&&0.122x^6 + 0.085x^{12},\\\nonumber
	\rho(x)&=&0.85x^5 + 0.15x^6.
\end{eqnarray}
The associated EPR-LDPC code with block length-2048 bits is optimized by the algorithm in section~\ref{sub:Profiles_optimization}. Then we give the performance comparison under different decoders in Fig~\ref{fig:Decoders_moderate}.
In this example, SEPR also outperforms SER and is more close to the QSPA than SER.
\begin{figure}[!hbtp]
\begin{center}
\includegraphics[width=.48\textwidth]{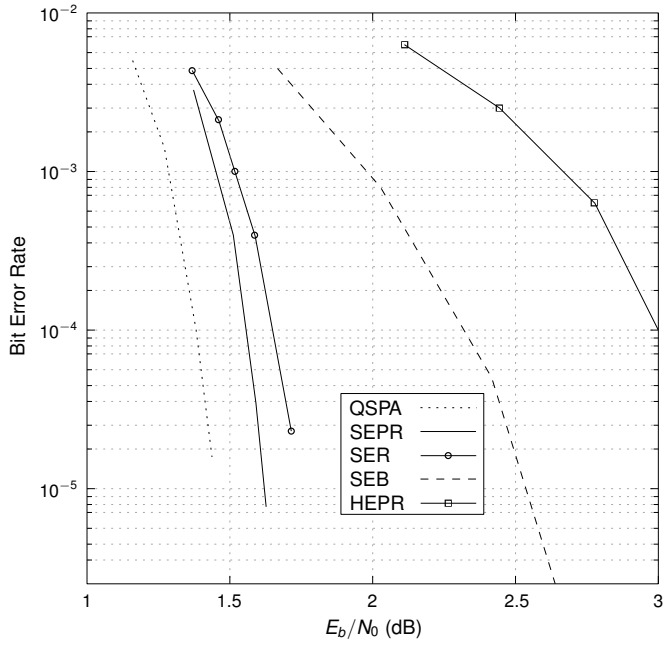}
\end{center}
\vspace*{-.3cm}
\caption{Performance comparison between different representations. Maximum 40 iterations, $\mu=16$ and $\nu=4$.}
\label{fig:Decoders}\vspace*{-.3cm}
\end{figure}

\begin{figure}[!hbtp]
\begin{center}
\includegraphics[width=.48\textwidth]{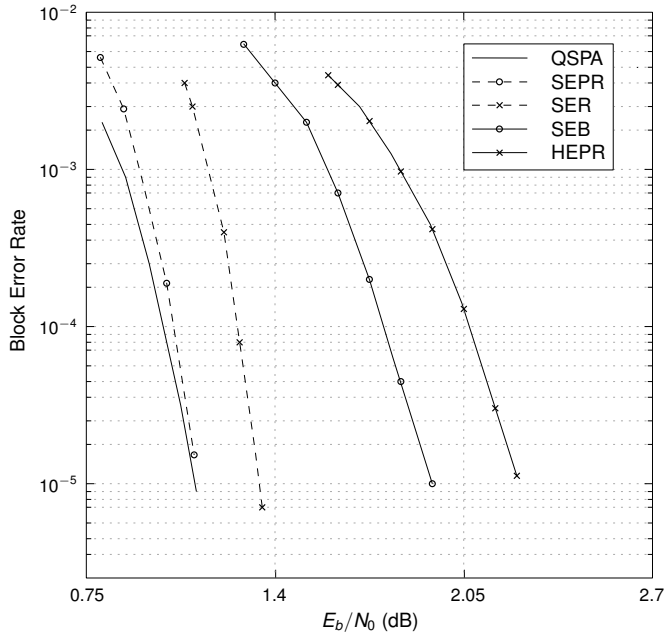}
\end{center}
\vspace*{-.3cm}
\caption{Performance comparison between different representations. Maximum 200 iterations, $\mu=16$ and $\nu=4$.}
\label{fig:Decoders_moderate}\vspace*{-.3cm}
\end{figure}
\subsection{Short length optimization} 
\label{sub:Finite_length}
Short length cycles can cause severe performance degradation for short length block codes.
We eliminate the short length cycles for $\mathbf\Omega^e$ in this subsection and give the comparative results for different outputs of the optimization in section~\ref{sub:Profiles_optimization}  which are displayed in table~\ref{tab:Algo_outputs}.
Let $\bar{\mathbf x}$ be the bit sequence transmitted over the binary input Gaussian channels.
${\mathbf H}$ is constructed by modified progressive edge growth (PEG) method and the hybrid parallel decoder is adopted. Let the block length be 360 bits.
Let $M_s=\sum_j\mathrm{wt}(\mathbf\Psi^e_j)$ be the length of $\mathbf v^e$ and $g_s$ be the girth.
If the non-binary (3,6)-regular LDPC code is adopted, we give the performance comparison in Fig~\ref{fig:Outputs}.
The 32-ary LDPC code achieves the best performance in our simulation due to the optimization both on the girth and field size.

\begin{table}[tb]
  \caption{Different outputs from Section~\ref{sub:Profiles_optimization}. $q$ is the field size, $g_s$ is the girth, $M_s$ is the length of $\mathbf v^e_j$ and $\frac{q}{2}-1$ is the sufficient condition for the successful decoding from Theorem~\ref{them:resolvable}.}
  \label{tab:Algo_outputs}
    \begin{center}
      \begin{tabular}{ccccc}
      \toprule
      $\mathbf v^e_j$ & $q=2^p$ & $g_s$ & $M_s$ & $\frac{q}{2}-1$ \\ \hline
      \multirow{2}*{$\mathbf v^e_j\preccurlyeq \mathbf v_j$} & 8 & 6 & 734 & 3 \\ \cmidrule(r){2-5}
      & 16 & 8 & 1415 & 7\\ \hline
      \multirow{2}*{$\mathbf v^e_j\prec \mathbf v_j$} & 8 & 6 & 507 & 3\\
      & 32 & 10 & 2132 & 15 \\
      \bottomrule
    \end{tabular}
  \end{center}
\end{table}

\begin{figure}[!hbtp]
\begin{center}
\includegraphics[width=.48\textwidth]{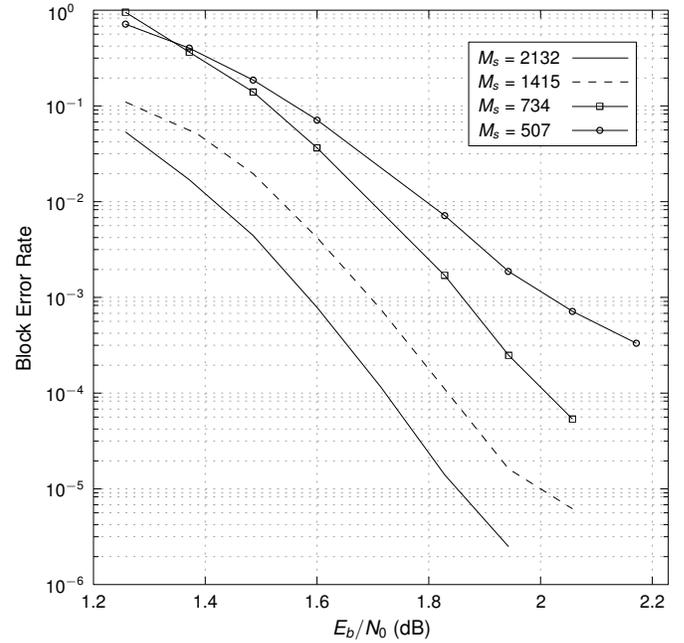}
\end{center}
\vspace*{-.3cm}
\caption{Performance comparison between different outputs in Table~\ref{tab:Algo_outputs}. }
\label{fig:Outputs}\vspace*{-.3cm}
\end{figure}
\subsection{Comparison of codes from literature} 
\label{sub:Codes_literature}
Consider the non-binary LDPC code of rate half over $\mathbf F_{16}$ in Section~\ref{sub:Different_binary_forms}.
We compare the performance of the EPR-LDPC code with the optimized non-binary cycle LDPC codes (optimized under similar assumptions) and the girth optimized binary LDPC codes in the literature.
In Fig~\ref{fig:Codes_lit}, SPB59 is the sphere packing bound for block length-2048 bits.
The codes from\cite{Huang09Cycle} is the non-binary cycle code with length 5376 bits.
The code from\cite{Poulliat08BinImage} is the non-binary cycle code with length 2048 bits.
The code from\cite{Voicila2008SplitCode} is the non-binary cycle code with length 3000 bits.
These codes are optimized for non-binary decoders.
The code from\cite{Christian2009QcGirth} is the (3,6) QC-LDPC code with length 2294 bits.
The code from\cite{Guohua2010QcGirth} is the PEG-LDPC code with length 2694 bits.
These codes are optimized for binary decoders.
Our irregular EPR-LDPC code with block length-2048 bits (decoded by the decoder in example~\ref{exam:Gaussian})
outperforms others even with much shorter length than the codes in \cite{Huang09Cycle}
and much smaller field size than the code in \cite{Poulliat08BinImage}.
The EPR-LDPC code has achieved a maximum $0.8$dB (at BER=$10^{-4}$) performance gain compared to the optimized non-binary cycle LDPC codes with a much lower decoding complexity.
\subsection{Binary Erasure Channel} 
\label{sub:Sim_Bec}
In this subsection, we investigate the performance of the EPR-LDPC code
over BEC with different size of matrix labels.
Consider the $p$-reducible code of rate half. Let $\mathbf\Lambda_{p_2}$ be a $500\times 1000$ mother matrix.
The degree distributions are displayed in table~\ref{tab:DegreeDist_2}.
The hybrid parallel decoder in Example~\ref{exam:Gaussian} is altered for BEC and the MSNR is $0.49$.
If $p=2,3,4,5,6$ and $\mathbf\Lambda_{p_2}$ is filled with randomly generated matrix labels, we compare the performance for different $\mathbf\Omega$ in Fig~\ref{fig:EPR_BEC}.
As the block length ($N_q=1000(q-1)$) increases, the decoding curves approach the MSNR as we expect.

\begin{figure}[!hbtp]
\begin{center}
\includegraphics[width=.48\textwidth]{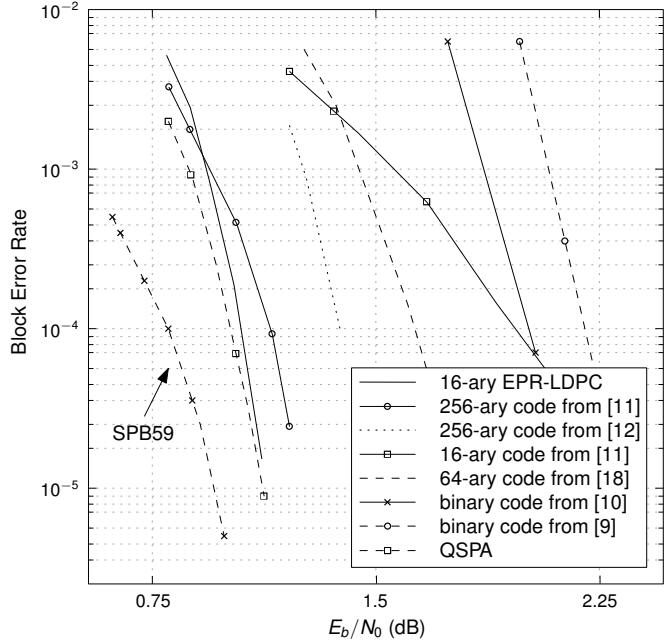}
\end{center}
\vspace*{-.3cm}
\caption{The EPR-LDPC code compared with codes from literature. }
\label{fig:Codes_lit}\vspace*{-.3cm}
\end{figure}

\begin{figure}[!hbtp]
\begin{center}
\includegraphics[width=.48\textwidth]{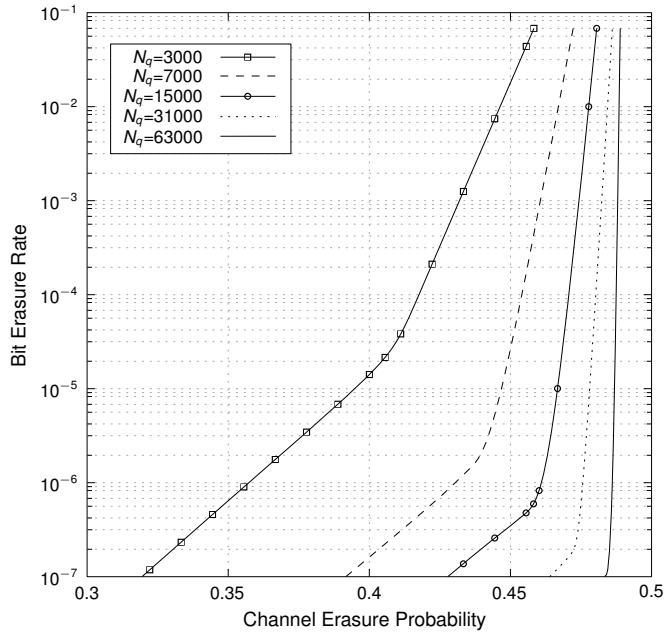}
\end{center}
\vspace*{-.3cm}
\caption{EPR-LDPC code over BEC with different size of matrix labels. $N_q$ is the block length in bits. }
\label{fig:EPR_BEC}\vspace*{-.3cm}
\end{figure}
\begin{table}[tb]
  \caption{MSNRs for different degree distributions.}
  \label{tab:DegreeDist_1}
    \begin{center}
      \begin{tabular}{cc | cc}
        \toprule
        \multicolumn{2}{c}{$\mathbf\Lambda_{p_1}$} & \multicolumn{2}{c}{$\mathbf H$} \\ \hline
        $\lambda(x)$ & $\rho(x)$ & $\lambda(x)$ & $\rho(x)$ \\ \hline \hline
0.183$x$      &  0.047$x^{3}$ & 0.153$x$      & 0.010$x^4$    \\
0.275$x^2$    &  0.186$x^{4}$ & 0.261$x^{2}$  & 0.029$x^5$    \\
0.096$x^3$    &  0.308$x^{5}$ & 0.138$x^{3}$  & 0.074$x^6$    \\
0.024$x^4$    &  0.277$x^{6}$ & 0.051$x^{4}$  & 0.178$x^7$    \\
0.049$x^5$    &  0.140$x^{7}$ & 0.047$x^{5}$  & 0.272$x^8$    \\
0.049$x^6$    &  0.038$x^{8}$ & 0.046$x^{6}$  & 0.248$x^9$    \\
0.026$x^7$    &  0.004$x^{9}$ & 0.026$x^{7}$  & 0.136$x^{10}$ \\
0.007$x^8$    &               & 0.007$x^{8}$  & 0.044$x^{11}$ \\
0.001$x^9$    &               & 0.001$x^{9}$  & 0.008$x^{12}$ \\
0.002$x^{16}$ &               & 0.001$x^{17}$ &               \\
0.008$x^{17}$ &               & 0.004$x^{18}$ &               \\
0.019$x^{18}$ &               & 0.012$x^{19}$ &               \\
0.032$x^{19}$ &               & 0.021$x^{20}$ &               \\
0.041$x^{20}$ &               & 0.029$x^{21}$ &               \\
0.039$x^{21}$ &               & 0.031$x^{22}$ &               \\
0.030$x^{22}$ &               & 0.026$x^{23}$ &               \\
0.020$x^{23}$ &               & 0.019$x^{24}$ &               \\
0.014$x^{24}$ &               & 0.015$x^{25}$ &               \\
0.014$x^{25}$ &               & 0.015$x^{26}$ &               \\
0.016$x^{26}$ &               & 0.019$x^{27}$ &               \\
0.017$x^{27}$ &               & 0.021$x^{28}$ &               \\
0.015$x^{28}$ &               & 0.020$x^{29}$ &               \\
0.011$x^{29}$ &               & 0.016$x^{30}$ &               \\
0.007$x^{30}$ &               & 0.011$x^{31}$ &               \\
0.004$x^{31}$ &               & 0.006$x^{32}$ &               \\
0.002$x^{32}$ &               & 0.003$x^{33}$ &               \\
0.001$x^{33}$ &               & 0.001$x^{34}$ &               \\
        \hline \hline
$T_b$ & -0.18dB & \multicolumn{2}{c}{0.59dB} \\ \hline
        \bottomrule
      \end{tabular}
    \end{center}
\end{table}

\begin{table}[tb]
  \caption{MSNRs for different degree distributions.}
  \label{tab:DegreeDist_2}
    \begin{center}
      \begin{tabular}{cc | cc}
        \toprule
        \multicolumn{2}{c}{$\mathbf\Omega^e$} & \multicolumn{2}{c}{$\mathbf\Lambda_{p_2}$} \\ \hline
        $\lambda(x)$ & $\rho(x)$ & $\lambda(x)$ & $\rho(x)$ \\ \hline \hline
  0.138$x$       &   0.002$x^{3}$  &   0.102$x$       &   0.02$x^{7}$ \\
  0.235$x^{2}$   &   0.004$x^{4}$  &   0.183$x^{2}$   &   0.095$x^{8}$ \\
  0.140$x^{3}$   &   0.007$x^{5}$  &   0.113$x^{3}$   &   0.205$x^{9}$ \\
  0.084$x^{4}$   &   0.051$x^{6}$  &   0.039$x^{4}$   &   0.260$x^{10}$ \\
  0.075$x^{5}$   &   0.176$x^{7}$  &   0.016$x^{5}$   &   0.218$x^{11}$ \\
  0.052$x^{6}$   &   0.291$x^{8}$  &   0.028$x^{6}$   &   0.128$x^{12}$ \\
  0.024$x^{7}$   &   0.268$x^{9}$  &   0.040$x^{7}$   &   0.054$x^{13}$ \\
  0.006$x^{8}$   &   0.147$x^{10}$ &   0.033$x^{8}$   &   0.016$x^{14}$ \\
  0.001$x^{9}$   &   0.048$x^{11}$ &   0.026$x^{9}$   &   0.003$x^{15}$ \\
  0.001$x^{13}$  &   0.006$x^{12}$ &   0.032$x^{10}$ \\
  0.003$x^{14}$  &                 &   0.038$x^{11}$ \\
  0.005$x^{15}$  &                 &   0.030$x^{12}$ \\
  0.008$x^{16}$  &                 &   0.016$x^{13}$ \\
  0.010$x^{17}$  &                 &   0.006$x^{14}$ \\
  0.013$x^{18}$  &                 &   0.001$x^{15}$ \\
  0.016$x^{19}$  &                 &   0.001$x^{59}$ \\
  0.020$x^{20}$  &                 &   0.001$x^{60}$ \\
  0.021$x^{21}$  &                 &   0.003$x^{61}$ \\
  0.021$x^{22}$  &                 &   0.006$x^{62}$ \\
  0.019$x^{23}$  &                 &   0.010$x^{63}$ \\
  0.017$x^{24}$  &                 &   0.015$x^{64}$ \\
  0.016$x^{25}$  &                 &   0.021$x^{65}$ \\
  0.016$x^{26}$  &                 &   0.027$x^{66}$ \\
  0.016$x^{27}$  &                 &   0.032$x^{67}$ \\
  0.014$x^{28}$  &                 &   0.035$x^{68}$ \\
  0.011$x^{29}$  &                 &   0.034$x^{69}$ \\
  0.008$x^{30}$  &                 &   0.031$x^{70}$ \\
  0.005$x^{31}$  &                 &   0.026$x^{71}$ \\
  0.002$x^{32}$  &                 &   0.020$x^{72}$ \\
  0.001$x^{33}$  &                 &   0.014$x^{73}$ \\
                 &                 &   0.009$x^{74}$ \\
                 &                 &   0.006$x^{75}$ \\
                 &                 &   0.003$x^{76}$ \\
                 &                 &   0.002$x^{77}$ \\
                 &                 &   0.001$x^{78}$ \\ \hline \hline
$T_b$ & 0.73dB & \multicolumn{2}{c}{0.49} \\ \hline
        \bottomrule
      \end{tabular}
    \end{center}
\end{table}
\section{Conclusion}
When there is no symbol-level cycle, the EPR-LDPC code will not have any bit-level cycle. Superior to the extended binary representation, the parity check matrix of EPR-LDPC code will not be composed of disjoint sub-matrices too. When there exists short length symbol-level cycles, the EPR-LDPC code can largely avoid the corresponding short length bit-level cycles.
Decoding of the EPR-LDPC code by the proposed decoders (the hybrid hard-decision decoder and the hybrid parallel decoder) is capable of achieving computational complexities of $O(m_s)$ where $m_s<q$.
Simulations show that the EPR-LDPC code outperforms the extended binary representation with the same block length.
In addition, compared to the optimized non-binary cycle LDPC codes under non-binary decoders, the EPR-LDPC code under the proposed decoder achieves a maximum $0.8$dB performance gain.

\ifCLASSOPTIONcaptionsoff
  \newpage
\fi

\bibliographystyle{IEEEtran}
\bibliography{IEEEabrv,QldpcEnDec}

\end{document}